\documentclass[a4paper]{amsart}

\usepackage{amssymb,amscd}
\usepackage{amsmath,amsfonts,latexsym}
\usepackage[dvips]{graphicx}
\usepackage{color}
\usepackage[english]{babel}
\usepackage[latin1]{inputenc}
\usepackage{enumerate,afterpage}
\usepackage{tikz}
\usepackage{url}
\usetikzlibrary{matrix,arrows}

\usepackage[all]{xy}

\usepackage{mathrsfs}
\usepackage{mathtools}
\usepackage{stmaryrd}
\usepackage{caption}[2003/12/20]

\newcounter{notes}%
\newcommand{\marginnote}[1]{
\refstepcounter{notes}  
\nolinebreak
$\hspace{-5pt}{}^{\text{\tiny \rm \arabic{notes}}}$
\marginpar{\tiny \arabic{notes}) #1}
}

\newtheorem{cor}{Corollary}[section]
\newtheorem{theorem}[cor]{Theorem}

\newtheorem{lemma}[cor]{Lemma}

\newtheorem{Theorem}{Theorem}[section]
\newtheorem{Lemma}[cor]{Lemma}
\newtheorem{Proposition}[cor]{Proposition}

\newtheorem{introthm}{Theorem}

\theoremstyle{definition}

 \newtheorem{Definition}[cor]{Definition}

\newtheorem{Remark}[cor]{Remark}
\newtheorem{remark}[cor]{Remark}

\def\co{\colon\thinspace} 

\newcommand{\Del}{\mathrm{Del}}

\newcommand{\aff}{\mathrm{Aff}}
\newcommand{\CH}{\mathrm{CH}}

\newcommand{\D}{\mathrm{d}}

\newcommand{\cC}{{\mathcal C}}

\newcommand{\HH}{{\mathbb H}}

\newcommand{\HP}{{\mathbb{HP}}}
\newcommand{\R}{{\mathbb R}}

\newcommand{\X}{{\mathbb X}}

\newcommand{\AdS}{\mathbb{A}\mathrm{d}\mathbb{S}}

\def\co{\colon\thinspace}

\newcommand{\RR}{\mathbb R}

\newcommand{\RP}{\mathbb{RP}}

\newcommand{\OO}{\operatorname{O}}

\newcommand{\PO}{\operatorname{PO}}

\newcommand{\Vor}{\mathrm{Vor}}

\newcommand{\Isom}{\operatorname{Isom}}

\let\oldtocsection=\tocsection

\let\oldtocsubsection=\tocsubsection

\let\oldtocsubsubsection=\tocsubsubsection

\renewcommand{\tocsection}[2]{\hspace{0em}\oldtocsection{#1}{#2}}
\renewcommand{\tocsubsection}[2]{\hspace{1em}\oldtocsubsection{#1}{#2}}
\renewcommand{\tocsubsubsection}[2]{\hspace{2em}\oldtocsubsubsection{#1}{#2}}

\begin{document}

\title{Higher signature Delaunay decompositions}

\author{Jeffrey Danciger}
\address{Department of Mathematics, University of Texas at Austin}
\email{jdanciger@math.utexas.edu}
\urladdr{www.ma.utexas.edu/users/jdanciger}

\author{Sara Maloni}
\address{Department of Mathematics, Brown University}
\email{sara\_maloni@brown.edu}
\urladdr{http://www.math.brown.edu/$\sim$maloni}

\author{Jean-Marc Schlenker}
\address{Department of mathematics, University of Luxembourg}
\email{jean-marc.schlenker@uni.lu}
\urladdr{math.uni.lu/schlenker}

\thanks{Danciger was partially supported by grant DMS-1510254 and Maloni was partially supported by grant DMS-1506920, from the National Science Foundation. The authors also acknowledge support from U.S. National Science Foundation grants DMS-1107452, 1107263, 1107367 ``RNMS: GEometric structures And Representation varieties'' (the GEAR Network).}

\date{\today (v0)}

\begin{abstract}
A Delaunay decomposition is a cell decomposition in $\RR^d$ for which each cell is inscribed in a Euclidean ball which is empty of all other vertices. 
This article introduces a generalization of the Delaunay decomposition in which the Euclidean balls in the empty ball condition are replaced by other families of regions bounded by certain quadratic hypersurfaces. This generalized notion is adaptable to geometric contexts in which the natural space from which the point set is sampled is not Euclidean, but rather some other flat semi-Riemannian geometry, possibly with degenerate directions. 
We prove the existence and uniqueness of the decomposition and discuss some of its basic properties. In the case of dimension $d = 2$, we study the extent to which some of the well-known optimality properties of the Euclidean Delaunay triangulation generalize to the higher signature setting. In particular, we describe a higher signature generalization of a well-known description of Delaunay decompositions in terms of the intersection angles between the circumscribed circles.

\end{abstract}

\maketitle

\tableofcontents

\section{Introduction}
Let $X$ be a finite set of points in Euclidean space $\RR^d$ of dimension $d \geq 2$. A cell decomposition of the convex hull $\CH(X)$ of $X$ is called a \emph{Delaunay decomposition} for $X$ if it satisfies the \emph{empty ball condition:} each cell of the decomposition is inscribed in a unique round ball which does not contain any points of $X$ in its interior. A classical theorem of Delaunay \cite{del_sur} states that if $X$ is in generic position (meaning that no $d+1$ points lie in an affine hyperplane and no $d+2$ points lie on any sphere), then there exists a unique Delaunay decomposition for $X$, all of the cells of which are simplices. 

This article introduces a generalization of the Delaunay decomposition in which the balls in the empty ball condition above are replaced by other families of regions bounded by certain quadratic hypersurfaces. This generalized notion of Delaunay decomposition is adaptable to geometric contexts in which the natural space from which the point set $X$ is sampled is not Euclidean, but rather some other flat semi-Riemannian geometry, possibly with degenerate directions. For example, in the Minkowski plane $\RR^{1,1}$, we replace Euclidean circles with the hyperbolas of constant Minkowski distance from a point. In the degenerate plane $\RR^{1,0,1}$, thought of as a rescaled limit of the Euclidean plane in which one direction has become infinitesimal, we replace Euclidean circles with their rescaled limits, namely parabolas.

A quadratic form $Q$ on $\RR^d$ defines a natural geometric structure on $\RR^d$.  In the case that $Q$ is non-degenerate, 
we simply use $Q$ in place of the Euclidean norm to measure distances (possibly imaginary), angles, and other geometric quantities. The points within a constant $Q$-distance of a fixed point are called $Q$-balls and we substitute these in place of Euclidean balls to define a natural notion of a $Q$--Delaunay triangulation. In the case that $Q$ has some degenerate directions, we study a natural geometry which treats the degenerate directions as infinitesimal and define an appropriate notion of $Q$-ball which is slightly more subtle than in the non-degenerate case. Again we define a Delaunay decomposition in terms of an appropriate empty ball condition. We stress that in the case that $Q$ is not positive (or negative) definite, $Q$-Delaunay decompositions are not in general obtained via projective transformations from Euclidean ones.

In Section~\ref{classic} we extend a classical proof of the existence and uniqueness of Delaunay decompositions in the Euclidean case to the current setting under the necessary assumption that the point set is in {\em spacelike position} with respect to $Q$, meaning that any two points are positive real distance apart.

\begin{introthm}\label{thm:Del-gen}
Let $Q$ be a quadratic form on $\RR^d$ and let $X\subset \R^d$ be a finite set in spacelike position with respect to $Q$ and in generic position (meaning that no $d+1$ points lie in an affine hyperplane and no $d+2$ points lie on any $Q$--sphere). Then there exists a unique Delaunay decomposition for $X$. We will denote this decomposition by $\Del_{Q}(X)$
\end{introthm}

In Section \ref{quad} we will review some details on quadratic forms on $\R^d$ and their associated geometry. In the case that the quadratic form $Q$ is not degenerate, $\Del_{Q}(X)$ is naturally dual to a generalized Voronoi decomposition, the vertices of which are the centers of the empty $Q$--balls of $\Del_Q(X)$ (Proposition \ref{Q-Voronoi}).
The case that $Q$ is degenerate is more delicate. In this case it is natural to split $\RR^d = \RR^m \oplus \RR^{d-m}$ as the sum of a non-degenerate subspace $\RR^m$ complementary to the degenerate subspace $\RR^{d-m}$. The associated geometry that we study treats the degenerate directions as infinitesimal so that a point $(p_0,v)$ in this geometry is thought as a point $p_0$ in the non-degenerate subspace $\RR^m$ plus infinitesimal information $v \in \RR^{d-m}$ describing the normal derivative of a path $p_t$ leaving the subspace $\RR^m$ going into $\RR^d$. We show in Theorem \ref{thm:rescaled} that the Delaunay decomposition $\Del_Q(X)$ in this case predicts for short time $t > 0$ the combinatorics of the path of Delaunay decompositions $\Del_{Q'}(X_t)$ for points $X_t$ whose first order data agrees with that determined by $X$ and with respect to a non-degenerate quadratic form $Q'$ obtained by making the infinitesimal directions finite.
 
In Section \ref{d2}, we study the extent to which some of the well-known optimality properties of the Euclidean Delaunay triangulation generalize to the non-Euclidean setting in the case of dimension $d = 2$. There are essentially two non-Euclidean quadratic forms $Q$, one Lorentzian (signature $(1,1)$) and one degenerate. In the first case, the $Q$-circles are hyperbolas and in the second case they are parabolas. 
As in the Euclidean case, in either of these geometries weights may be assigned to the edges of a triangulation by summing the angles at opposite vertices, or alternatively by measuring the angle of intersection of the circumscribed $Q$--circles containing a given edge. Using recent results~\cite{dan_pol} about three-dimensional ideal polyhedra in constant negative curvature geometries, in Theorem \ref{thm:prescribe-angles} we characterize the possible edge weights that may occur for a $Q$--Delaunay triangulation in either case. In addition, using a generalized Thales' Theorem, Theorem \ref{thm:angle-optimization} proves that $Q$--Delaunay triangulations are optimally fat: they maximize the angle sequence (ordered lexicographically) over all triangulations as in the Euclidean case. 
 
In Section \ref{sc:5}, we briefly outline some possible applications of the higher signature 
Delaunay decompositions developed here, and mention a few open questions.

\section{Classical Delaunay decomposition and its generalizations}\label{classic}

Recall that a ball is defined in terms of the Euclidean norm, $\|x\|^2 = x^Tx$. Specifically, the ball of radius-squared $D\in \RR$ and center $p \in \RR^d$ is the set of points $x \in \RR^d$ satisfying the inequality:
\begin{align}
\|x - p\|^2 &\leq D \label{eqn:ball1}.
\end{align}
This may be expressed as the equivalent condition that:
\begin{align}
\|x\|^2 &\leq \varphi(x) + D', \label{eqn:ball2}
\end{align}
where $\varphi: \RR^d \to \RR$ is the linear functional defined by $\varphi(x) = 2p^Tx$ and the constant $D' = D - \|p\|^2$. In other words, a ball is exactly the set of points in $\RR^d$ whose image in $\RR^{d+1}$ on the graph of the function $\|\cdot \|^2$ lies below some affine hyperplane.
The norm-squared function $x \mapsto \|x\|^2 = x^Tx$ is one example of a \emph{quadratic form}, which we define here as any function $Q: \RR^d \to \RR$ of the form $$Q(x) = x^T A x,$$ where $A$ is some symmetric $d \times d$ matrix. If the matrix $A$ is non-singular, then we call $Q$ {\it non-degenerate}, and if the matrix $A$ is positive definite (respectively indefinite), then we call $Q$ {\it positive
definite} (respectively {\it indefinite}). Up to transforming $\RR^d$ by a linear isomorphism, a quadratic form $Q$ is determined entirely by the number of positive, negative, and zero eigenvalues of~$A$. This data is referred to as the signature of $Q$.

Let $Q$ be a non-degenerate quadratic form, possibly indefinite, on $\RR^d$. The term \emph{$Q$--ball} refers to any region defined by the inequality 
\begin{align}\label{eqn:Qball1}
Q(x-p) \leq D,
\end{align} 
where $p \in \RR^d$ and $D \in \RR$ (possibly negative). Equivalently, a $Q$--ball is any region defined by the inequality
\begin{align}\label{eqn:Qball2}
Q(x) \leq \varphi(x) + D',
\end{align}
where $\varphi: \RR^d \to \RR$ is any linear functional and $D' \in \RR$ is any constant.

Of course, Inequality~\eqref{eqn:Qball1} may be described in the form~\eqref{eqn:Qball2} by simply taking $\varphi(x) = 2p^TAx$ and $D' = D + Q(p)$. However, solving for $p, D$ in terms of $\varphi, D'$ in order to deduce inequality~\eqref{eqn:Qball1} from~\eqref{eqn:Qball2} 
requires the quadratic form $Q$ (that is, the matrix $A$) to be non-degenerate. Hence, Inequalities~\eqref{eqn:Qball1} and~\eqref{eqn:Qball2} define different families of regions in the case that $Q$ is degenerate. It turns out that the regions defined by Inequality~\eqref{eqn:Qball2} define a much more interesting theory of Delaunay decomposition in this case.

\begin{Definition}[$Q$--ball]\label{def:Qball}
Let $Q$ be any quadratic form (possibly degenerate) on $\RR^d$. The term \emph{$Q$--ball} refers to the region defined by Inequality~\eqref{eqn:Qball2}, where $\varphi \co \RR^d \to \RR$ is a linear functional and $D' \in \RR$ is a constant. The boundary of a $Q$--ball is called a \emph{$Q$--sphere}.
\end{Definition}

\noindent We will see in Section~\ref{sec:degen} that a degenerate form $Q$ defines a special geometry in which the degenerate directions of $Q$ are treated as infinitesimal and we think of $Q$ as a degeneration of a family of non-degenerate quadratic forms $Q_t'$ as $t \to 0$. Under that interpretation, $Q$--balls may be thought of as limits of $Q'_t$-balls with center escaping to infinity as $t \to 0$.

In order to obtain a sensible theory of Delaunay decompositions using a semi-definite quadratic form $Q$, we must restrict the point sets under consideration. We make the following definition:

\begin{Definition}
Let $Q$ be a quadratic on $\RR^d$ (possibly indefinite and/or degenerate). A finite set of points $X$ in $\RR^d$ is said to be in \emph{spacelike position} with respect to $Q$ if, for any distinct points $x,y \in X$, the positivity condition $Q(x - y) > 0$ holds.
\end{Definition}

Replacing the notion of ball with that of $Q$--ball in the definition of the classical Delaunay decomposition yields the following generalization.

\begin{Definition}[$Q$--Delaunay decomposition]\label{def:Del-gen}
Let $Q$ be any quadratic form on $\RR^d$ and let $X$ be a finite set of points in $\RR^d$ in spacelike position with respect to $Q$. A cell decomposition of the convex hull $\CH(X)$ of $X$ with $0$-skeleton $X$ is called a \emph{$Q$--Delaunay decomposition} if it satisfies the \emph{empty $Q$--ball condition:}  each cell of the decomposition is inscribed in a unique $Q$--ball which does not contain any points of $X$ in its interior. We will call the boundary of such a $Q$--ball an {\em empty $Q$--sphere}.
\end{Definition}

\begin{figure}
\includegraphics[width = 5.0in]{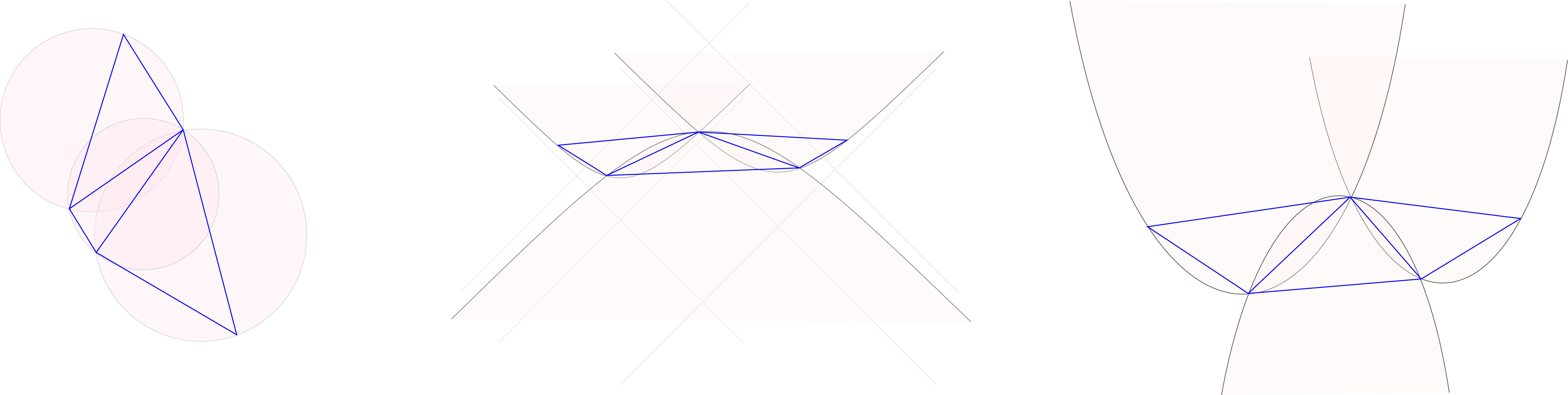}
\captionsetup{singlelinecheck=off}\caption[An example of the $Q$--Delaunay decomposition of five points in spacelike general position in the case that:]{An example of the $Q$--Delaunay decomposition of five points in spacelike general position in the case that:
\begin{itemize}
  \item (left) $Q(x_1,x_2) = x_1^2 + x_2^2$ is the standard Euclidean norm, and so the $Q$--balls are bounded by circles;
  \item (middle) $Q(x_1, x_2) = x_1^2 -x_2^2$ is the standard Minkowski norm, and so the $Q$--balls are bounded by hyperbolas;
  \item (right) $Q(x_1,x_2) = x_1^2$ is a degenerate norm, in which case the $Q$--balls are bounded by parabolas.
\end{itemize}
}
\end{figure}

\begin{figure}
\centering
\begin{minipage}{0.45\textwidth}
\centering
\includegraphics[width = 7cm]{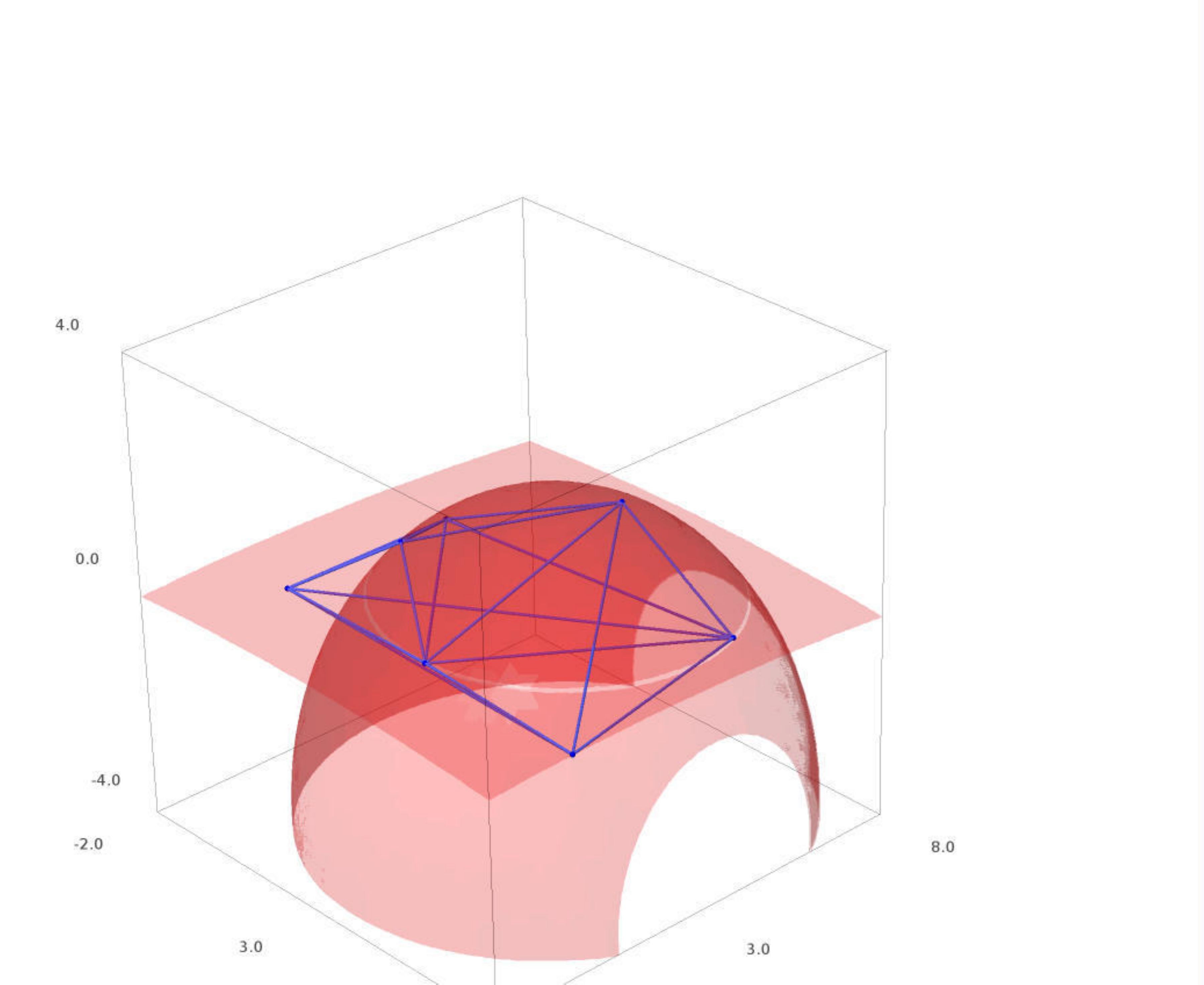}
\end{minipage}\hfill
\begin{minipage}{0.45\textwidth}
\centering
\includegraphics[width = 7cm]{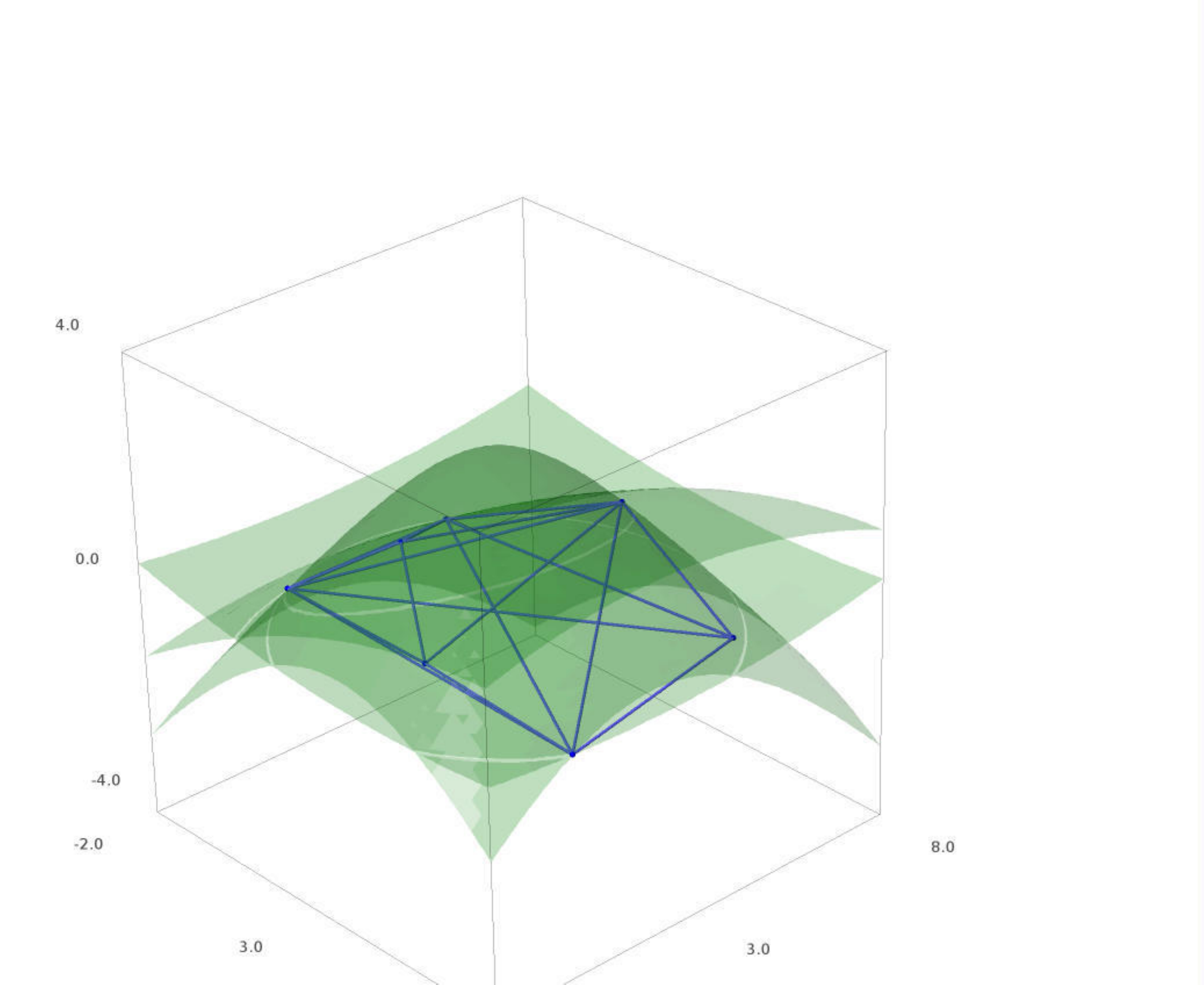}
\end{minipage}
\caption{The $Q$-Delaunay decompositions of the same 7 points when $Q$ is the Euclidean
(left) and the Minkowski (right) bilinear form on $\R^3$. Only some of the $Q$-balls 
have been drawn.}
\end{figure}

We now prove Theorem \ref{thm:Del-gen}. The proof is adapted from a well-known proof in the case that $Q$ is the standard Euclidean norm, see \cite{bro_vor}.

\begin{proof}[Proof of Theorem \ref{thm:Del-gen}]
  Given the quadratic form $Q$ on $\R^d$, we consider the quadratic form $\widehat Q$ on $\RR^{d+2}$ defined by $$\widehat Q(x_1, \ldots, x_d, x_{d+1}, x_{d+2}) := Q(x_1, \ldots, x_d) -x_{d+1}x_{d+2}.$$
  Associate to $\widehat Q$, we define the open subset $\mathbb X \subset \mathbb{RP}^{d+1}$ as follows:
  $$\mathbb X = \{x \in \mathbb{R}^{d+2}\setminus\{0\} \mid \widehat Q(x) < 0\}/\R^*,$$ 
  that is the set of negative lines in $\RR^{d+2}$ with respect to $\widehat Q$. Then $\mathbb X$ is a model for semi-Riemannian geometry, possibly with degenerate directions, associated with $\widehat Q$ with constant negative sectional curvature. In the affine chart $x_{d+2} =1$, the ideal boundary of $\mathbb X$ is described by the equation $$x_{d+1} = Q(x_1, \ldots, x_d),$$ i.e. the graph $\mathscr P$ in $\RR^{d+1}$ of $Q$. The projection map $\varpi: \RR^{d+1} \to \RR^d$ onto the first $d$-coordinates, defined by $$\varpi(x, x_{d+1}) = x,$$ restricts to a homeomorphism $\mathscr P \to \RR^d$, which, in the case that $Q$ is the usual Euclidean norm-squared, is usually called \emph{stereographic projection}.

Now, consider the convex hull $\mathscr C$ of the inverse images $X' = \varpi^{-1}(X)$ of $X$ on~$\mathscr P$. $\mathscr C$ may be regarded as a convex polyhedron in $\mathbb X$ whose vertices lie on the ideal boundary $\partial \mathbb X$. By the genericity assumption, $\mathscr C$ is a $d+1$-dimensional polyhedron, whose boundary is a union of $d$-dimensional simplices which come in two types, called \emph{bottom} and \emph{top}. A face $f$ on the boundary $\partial \mathscr C$ of $\mathscr C$ is a bottom face if its unit normal, directed toward the interior of $\mathscr C$, has positive $x_{d+1}$ component. A face $f$ is a top face if its inward directed unit normal has negative $x_{d+1}$ component. The unit normal to a face never has zero $x_{d+1}$ component by the genericity assumption.

\begin{figure}[h]
{
\centering
\def\svgwidth{6.0cm}
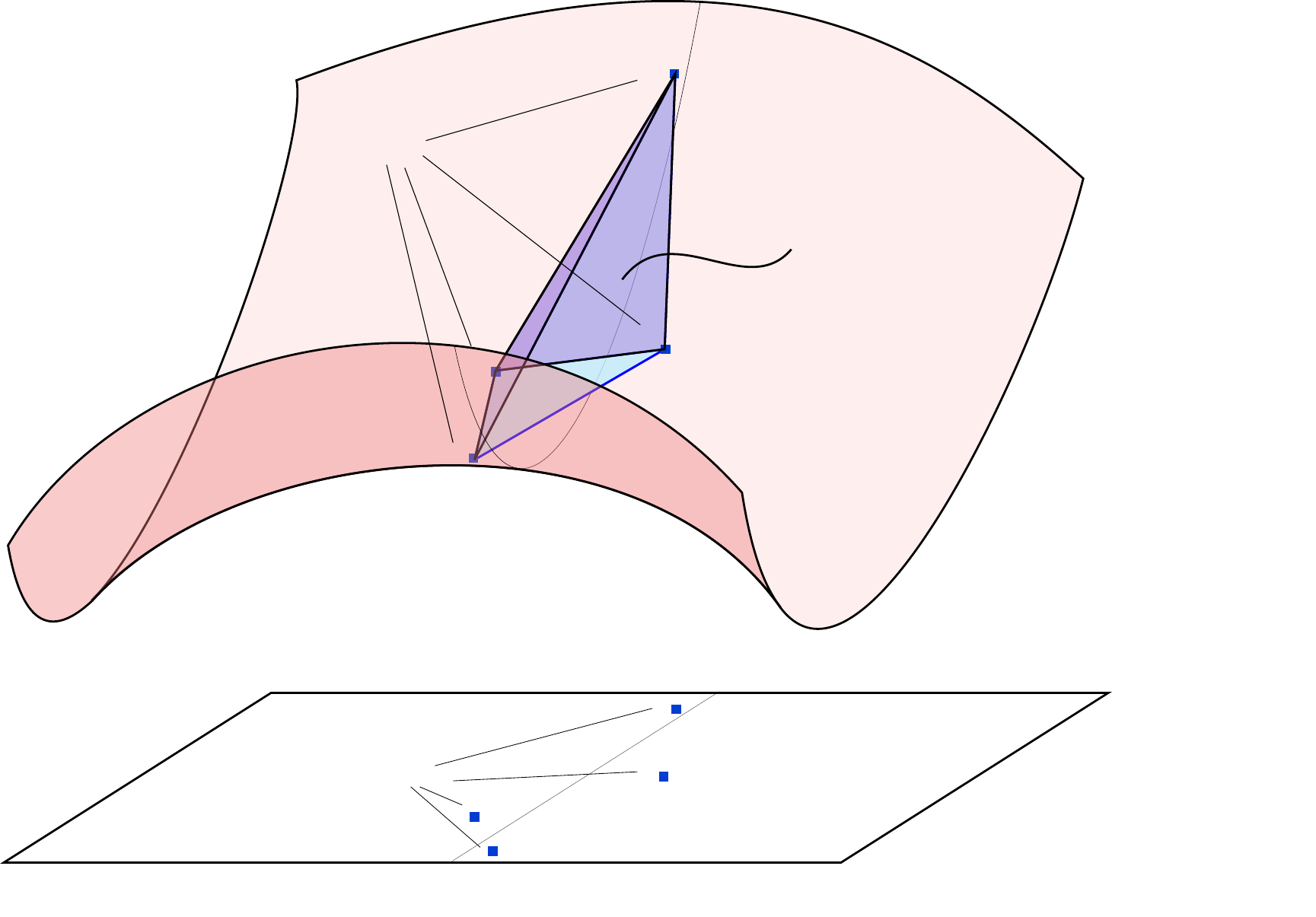
\def\svgwidth{6.0cm}
\input{projection-to-DelQ.pdf_tex}
}
\caption{To find the $Q$--Delaunay decomposition $\Del_Q(X)$, first lift $X$ to its image $X'$ on the graph of $Q$ and take the convex hull $\mathscr C$ of $X$ (left panel). Then project the bottom faces $\partial \mathscr C_-$ down to $\RR^d$ (right panel).}
\end{figure}

We claim that $\varpi$ restricts to a homeomorphism from the union $\partial \mathscr C_-$ of the bottom faces of $\partial \mathscr C$ to the convex hull $\CH(X)$ of $X$ in $\RR^d$ and that $\varpi$ takes the faces of $\partial \mathscr C_-$ to the cells of the unique $Q$--Delaunay decomposition for $X$. To see this, observe that the condition that all points $(x, Q(x))$ of $X'$ lie (weakly) above the affine hyperplane spanned by some face $f$ of $\partial \mathscr C_-$ is described by an inequality of the form
$$Q(x) \geq \varphi(x) + D'$$
for some linear functional $\varphi$ and constant $D'$, with equality exactly when $(x, Q(x))$ is a vertex of $f$. Hence, the cell decomposition of $\RR^d$ defined by projection of the cell structure of $\partial \mathscr C_-$ satisfies that the vertices of any face $\varpi(f)$ lie on a $Q$--sphere $S_f$ which bounds a $Q$--ball $B_f$ whose interior does not contain points of $X$. 
In the case that $Q$ is positive definite the proof, which up to this point is well-known, is complete.
There is, however, one subtle point to check if $Q$ is not positive definite, which is the case of interest in this article. Definition~\ref{def:Del-gen} requires that the simplex $\varpi(f)$ be \emph{inscribed} in $B_f$: so we must check that $\varpi(f)$ is actually contained in~$B_f$. This will follow from the assumption that $X$ is in spacelike position. Let $\langle \cdot, \cdot \rangle_Q $ denote the inner product associated to $Q$, defined by $\langle x, y\rangle_Q = x^T Ay$ for all $x,y \in \RR^d$. 
Then 
\begin{equation}\label{eqn:useful}
Q(x - y) > 0 \iff 2\langle x, y\rangle_Q < Q(x) +Q(y).
\end{equation}
 Let $x_1',\ldots,x_m'$ be vertices of the face $f$, where $x_k' = (x_k, Q(x_k)),$ for each $k \in \{1, \ldots, m\}$. A typical point of $f$ is a convex combination of the form $y = t_1x_1 + \cdots + t_mx_m$, where $t_1, \ldots, t_m$ are positive coefficients such that $t_1 + \cdots + t_m = 1$.
Then
\begin{align*}
Q(y) &= \left\langle \sum_{i=1}^m t_i x_i, \sum_{j=1}^m t_jx_j\right\rangle_Q \\
&= \sum_{i=1}^m t_i^2 Q(x_i) + 2\sum_{j < k} t_jt_k\langle x_j, x_k\rangle_Q\\
&\leq \sum_{i=1}^m t_i^2 Q(x_i) + \sum_{j < k} t_jt_k(Q(x_j) + Q(x_k))\\
& = \sum_{i=1}^m t_i(t_1 + \cdots +t_m) Q(x_i) = \sum_{i=1}^m t_i Q(x_i)\\
&= \sum_{i=1}^m t_i (\varphi(x_i) + D)  = \varphi(y) + D\\
&\implies Q( y) \leq \varphi(y) + D.
\end{align*}
The hypotheses that the points $x_1, \ldots, x_m$ are in spacelike position was applied (via~\eqref{eqn:useful}) in the third line. Note that equality holds if and only if $y = x_k$ for some $k \in \{1, \ldots, m\}$. It follows that $\varpi(f) \subset B_f$ with $\varpi(f) \cap S_f = \{x_1, \ldots, x_m\}$, that is $\varpi(f)$ is inscribed in $B_f$. The same ideas can be used to prove the uniqueness.
\end{proof}

\begin{remark}\label{full}
  Note that $\varpi$ restricts to a homeomorphism from the union $\partial \mathscr C_+$ of the top faces of $\partial \mathscr C$ to the convex hull $\CH(X)$ of $X$ in $\RR^d$ and that $\varpi$ takes the faces of $\partial \mathscr C_+$ to the cells of the unique decomposition for $X$ where each $Q$--ball satisfies the {\it full $Q$--ball condition}: each cell of the decomposition is inscribed in a unique $Q$--ball such that it contains all the points of $X$ in its closure. We will call the boundary of such a $Q$--ball a {\em full $Q$--sphere}.
\end{remark}

\begin{remark}
In the familiar case that $Q$ is the Euclidean norm, the surface $\mathscr P$ in the proof above may be naturally interpreted as the ideal boundary of the hyperbolic space $\HH^{d+1}$ with one point removed. The convex hull $\mathscr C$ of $X'$ is an ideal polyhedron in $\HH^{d+1}$, the geometry of which gives information about the geometry of the Delaunay decomposition of $X$. In the general case a similar interpretation is possible. If $Q$ is non-degenerate of signature $(p,q)$, then $\mathscr P$ identifies with (a dense subset of) the ideal boundary of the negatively curved semi-Riemannian space $\HH^{p+1,q}$ of signature $(p+1,q)$ and constant negative curvature. If $Q$ is degenerate, then $\mathscr P$ may be interpreted as (a dense subset of) the boundary of a degenerate geometry obtained from some $\HH^{p+1,q}$ as a limit. See Section~\ref{sec:Hpq}.
\end{remark}

\section{The geometry associated to a quadratic form}\label{quad}

\subsection{Non-degenerate quadratic forms}\label{sec:non-degen}

\subsubsection{Flat semi-Riemannian geometry}

Throughout this section, we assume $Q$ is a non-degenerate quadratic form on $\RR^d$. The numbers $p$ and $q$ of positive and negative  eigenvalues of the matrix $A$ of $Q$ 
satisfy that $p+q =d$. The pair $(p,q)$ is called the signature of $Q$. 
After changing coordinates by a linear transformation, we may assume that $Q$ is 
the standard quadratic form of signature $(p,q)$, meaning its matrix $A = I_p \oplus -I_q$ 
is simply the diagonal matrix with $p$ $(+1)$'s followed by $q$ $(-1)$'s on the diagonal. 
The inner product $\langle \cdot, \cdot \rangle_Q$ induced by $Q$ is given by the simple 
expression 
$$\langle x , y \rangle_Q = x_1y_1 + \cdots + x_p y_p - x_{p+1} y_{p+1} - \cdots -x_{d} y_{d}~.$$
We assume that $p,q>0$, since the case $q = 0$ is the classical Euclidean case and the case $p = 0$ is un-interesting (there are no point sets $X$ in spacelike position when $|X| > 1$). 

The inner product $\langle \cdot, \cdot \rangle_Q$ defines, by parallel translation, a flat semi-Riemannian metric on the affine space of $\RR^d$. This space and metric together is denoted by $\RR^{p,q}$. It is the simply connected model for flat semi-Riemannian geometry of signature $(p,q)$. The isometry group of $\RR^{p,q}$ is precisely $$\Isom(\RR^{p,q}) = \operatorname{Aff}(\OO(Q)) = \OO(p,q) \ltimes \RR^d,$$
where $\OO(Q) = \OO(p,q)$ denotes the linear transformations of $\RR^d$ preserving $Q$. 
Of course, any natural construction in flat semi-Riemannian geometry of signature $(p,q)$ should be invariant under this group. 

\begin{Proposition}\label{invariance1}
The set of $Q$--balls is invariant under $\Isom(\RR^{p,q})$, hence so is the notion of $Q$--Delaunay decomposition: if $f \in \Isom(\RR^{p,q})$ and $X \subset \RR^{p,q}$ is in spacelike and generic position, then $\Del_Q(f(X))=f(\Del_Q(X))$.
\end{Proposition}

\noindent In Section \ref{mobius}, we will show that in fact the $Q$--Delaunay decomposition is invariant under a larger set of transformations, specifically a natural subset (depending on $X$) of M\"obius transformations (Theorem~\ref{Mob}).

The quadratic form $Q$ may be used to measure lengths, just as in the Euclidean setting, with some minor differences.
First, note that the $Q$--distance-squared $Q(x-y)= \langle x-y, x-y \rangle_Q$ between two distinct points $x, y \in \RR^d$ may be positive, negative, or zero. Borrowing terminology from the Lorentzian setting ($p = d-1, q=1$), we call the displacement between $x$ and $y$ \emph{spacelike} if $Q(x - y) > 0$, \emph{timelike} if $Q(x-y) < 0$, or \emph{lightlike} if $Q(x-y) = 0$. When $Q(x-y) < 0$, one may think of the distance between $x$ and $y$ as an imaginary number, however it is more convenient to simply work with $Q(x-y)$ rather than its square root.

\subsubsection{$Q$--Voronoi decompositions}

As in the classical setting, the $Q$--Delaunay decomposition is dual to an analogue of the Voronoi decomposition, which we call the $Q$--Voronoi decomposition.

\begin{Definition}
Let $X \subset \RR^d$ be finite and let $Q$ be a non-degenerate quadratic form. Then for each $x \in X$, we define the Voronoi region of $x$ to be $$\Vor_Q(x) := \{y \in \RR^d : Q(x-y) \leq Q(x'-y) \text{ for all } x' \in X\}.$$
\end{Definition}

\begin{Proposition}\label{Q-Voronoi}
Suppose $X \subset \RR^d$ is a finite set in spacelike and generic position with respect to $Q$. Then 
\begin{enumerate}
\item For each $x \in X$, $\Vor_Q(x)$ is a (possibly unbounded) convex polyhedron,
with timelike codimension 1 faces, containing $x$ in its interior.
\item The collection of $\Vor_Q(x)$ for all $x \in X$ gives a polyhedral decomposition of $\RR^d$ which we call the \emph{$Q$--Voronoi decomposition} with respect to $X$. It will be denoted by $\Vor_Q(X)$.
\item For each facet $f$ of codimension $k$ of $\Vor_Q(X)$, there exists $k+1$ unique points $p_0, p_1, \ldots, p_k$, distinct from one another, such that $$f = \Vor_Q(p_0) \cap \Vor_Q(p_1) \cap \cdots \cap \Vor_Q(p_k).$$ 
\item The vertices of $\Vor_Q(X)$ are the centers of the empty $Q$--balls of $\Del_Q(X)$. 
\end{enumerate}
\end{Proposition}

\begin{remark}
  It is also natural to define the {\em inverse Voronoi region} as follows: for each $x \in X$, consider the set $$\Vor^{-1}_Q(x) = \{y \in \RR^d : Q(x-y) \geq Q(x'-y) \text{ for all } x' \in X\}.$$ Then, denoting by $\Vor^{-1}_Q(X)$ the collection of $\Vor^{-1}_Q(x)$ for all $x \in X$, the vertices of $\Vor^{-1}_Q(X)$ are the centers of the full $Q$--balls, as defined in Remark \ref{full}.
\end{remark}

\begin{proof}
Since $X$ is in spacelike position, for all distinct $x, p \in X$,
the set of points $y\in \RR^d$ such that $Q(x-y) \leq Q(p-y)$ is a half-space, bounded by a 
timelike hyperplane, containing $x$ in its interior. Therefore, $\Vor_Q(x)$ is the union of finitely many
half-spaces (each bounded by a timelike hyperplane) and it is therefore a convex polyhedron
(which might however be unbounded). Moreover $\Vor_Q(x)$ contains $x$ in its interior, and
the first point is proved.

The second point is clear since, by definition, each point $y\in \RR^d$
is contained in at least one of the $\Vor_Q(x)$ for some $x \in X$. 

The third point follows from the first point and from the hypothesis that 
$X$ is in generic position: given a codimension $k$ facet $f$ of $\Vor_Q(X)$ in the
boundary of $\Vor_Q(p_0)$, for some $p_0 \in X$, there are $k$
other pairwise distinct points $p_1, \cdots, p_k\in X\setminus\{p_0\}$ such that 
for all $y\in f$ and for all $j\in \{ 0,\cdots, k\}$, we have $Q(y-p_i)=Q(p'_j-y)$.
Moreover there cannot be more than $k$ such points, by the genericity 
hypothesis. 

The last point follows from the third point because a vertex $v$ is a codimension $d$ facet. Hence, $v$ is $Q$-equidistant to $d+1$ points of $X$ and is further away from the other points. 
\end{proof}

\begin{Remark}
Note that it is necessary to suppose that $X$ is in spacelike position to have $x \in \Vor_Q(x)$.
  The hypothesis  that $X$ is in generic position, however, is not as essential. 
  If it is not satisfied, the same convex hull construction naturally defines a $Q$--Delaunay decomposition some of whose cells may not be simplices. The codimension $0$
  non-simplicial cells of the decomposition then correspond to empty spheres containing 
  more than $d+1$ points of $X$. The same phenomenon  appears for the usual 
  Delaunay decomposition in Euclidean space.
\end{Remark}
 
\subsection{Degenerate quadratic forms}\label{sec:degen}

Consider a time-dependent finite point set $X_t$ in $\RR^d$. That is, for each $t \geq 0$, $X_t = \{p_1(t), \ldots, p_n(t)\}$, where
for each $k \in \{1, \ldots, n\}$, $p_k(t)$ is a smooth path. We assume that at time $t = 0$, all of the points lie in a proper subspace $V \subset \RR^d$ and are in generic position within that subspace. We are interested in the time-dependent Delaunay decomposition of $X_t$, with respect to some non-degenerate quadratic form $Q$ on $\RR^d$ (perhaps the standard Euclidean norm, or perhaps an indefinite form defining a flat semi-Riemannian structure as in the previous section). Of course, at time $t = 0$, the points $X_0$ are not in generic position and the $ Q$-Delaunay decomposition is not defined (degenerate). However, we wish to use infinitesimal data at time $t = 0$ to predict the combinatorics of the Delaunay decomposition for $t > 0$.

Assume henceforth that the subspace $V$ which contains the points at time $t=0$ is non-degenerate with respect to $Q$, and let $U$ denote the $Q$-orthogonal complement. Then, after changing coordinates appropriately, we may assume that $V = \RR^m$ is the coordinate hyperplane corresponding to the first $m$ coordinates and that $U = \RR^{d-m}$ is the coordinate hyperplane corresponding to the remaining $d-m$ coordinates, so that the direct sum decomposition $\RR^d = \RR^m \oplus \RR^{d-m}$ is $Q$--orthogonal.\marginnote{S: Is it ok?\\ JD: modified} In these coordinates we write $p_k(t) = (y_k(t), z_k(t))$ for all $1 \leq k \leq n$.
Consider the set $$X = \left\{ q_k = (y_k(0), z_k'(0)) \mid 1 \leq k \leq n \right\},$$ where $z_k'(0)$ denotes the derivative of $z_k(t)$ at $t = 0$. Let $\mathscr Q$ denote the degenerate quadratic form defined by 
$$\mathscr Q (y,z) = Q(y,0)~. $$

\begin{theorem}\label{thm:rescaled}
With notation as above, assume that $ X$ is in spacelike, generic position for~$ \mathscr Q$. Then for all $t> 0$ sufficiently small, the natural map $X_t \to X$ induces a cell-wise homeomorphism taking $\Del_{Q}(X_t)$ to $\Del_{\mathscr Q}( X)$. In other words, for short time, the combinatorics of the $Q$--Delaunay decomposition of $X_t$ agrees with the combinatorics of the  $\mathscr Q$-Delaunay decomposition of $X$.
\end{theorem}

The theorem will follow from the next lemma which gives a natural interpretation of $\mathscr Q$-balls as rescaled limits of $Q$--balls under the rescaling maps $L_t: \RR^d \to \RR^d$ defined by $$L_t(y,z) = \left(y, \frac{z}{t}\right).$$ The proof, omitted, is a simple computation.

\begin{lemma}\label{lem:converging-balls}
Let $B$ denote the $ \mathscr Q$-ball defined by the equation $ \mathscr Q(y,z) \leq \varphi(y,z) + D$ for some linear functional $\varphi$ and constant $D$. Let $B_t$ be a family of $Q$--balls, defined for $t > 0$, by the equations $Q(y,z) \leq \varphi_t(y,z) + D_t$ where
\begin{align*}
\lim_{t\to 0} \ \varphi_t\circ L_t^{-1} &= \varphi,\\ \lim_{t \to 0} D_t &= D.
\end{align*}
Then $L_t B_t \to B$ as $t \to 0$.
\end{lemma}

And now the proof of Theorem~\ref{thm:rescaled}. 

\begin{proof}[Proof of Theorem~\ref{thm:rescaled}]
Let $I$ denote the $d+1$ indices of the points of $X$ lying on the boundary of an empty $\mathscr Q$-ball $$B^I = \{(y,z) \mid \varphi^I(y_i,z_i) + D^I\}$$
of $\Del_{{\mathscr Q}}(X)$, where the (unique) linear functional $\varphi^I$ and constant $D^I$ are defined by $$\mathscr Q(y_i,z_i') = \varphi^I(y_i,z_i') + D^I \;\; \forall i \in I.$$ For every sufficiently small $t > 0$, let $\varphi^I_t$ and $D^I_t$ be the unique linear functional and constant defining the $Q$--ball $B^I_t$ containing $(y_i(t), z_i(t))$ in its boundary for all $i \in I$. Then
$$ Q(y_i(t),z_i(t)) = \varphi^I_t(y_i(t), z_i(t)) + D^I_t $$
and therefore
\begin{align*}
Q\left(y_i(t), t \frac{z_i(t)}{t}\right) & = \varphi^I_t\left(y_i(t), t\frac{z_i(t)}{t}\right) + D^I_t \\ &= \varphi^I_t \circ L_t^{-1}\left(y_i(t), \frac{z_i(t)}{t} \right) + D^I_t~.
\end{align*}
Taking the limit as $t \to 0$ we find that
$$ Q(y_i(0), 0) = \lim_{t \to 0} \left(\varphi^I_t \circ L_t^{-1}\right)(y_i(0), z'_i(0)) + \lim_{t \to 0} D^I_t~, $$
so that 
$$ \mathscr Q(y_i(0), z_i'(0)) = \lim_{t \to 0} \left(\varphi^I_t \circ L_t^{-1}\right)(y_i(0), z'_i(0)) + \lim_{t \to 0} D^I_t
$$
for all $i \in I$. It now follows that $\varphi^I_t \circ L_t^{-1} \to \varphi$ and $D^I_t \to D^I$. We now apply Lemma \ref{lem:converging-balls} to obtain that $L_t B^I_t \to B^I$ as $t \to 0$. Observing that $L_t X_t \to  X$, it now follows that for sufficiently small $t > 0$, the rescaled balls $L_t B^I_t$ are empty of the other points of the rescaled point set $L_t X_t$. Hence, for sufficiently small $t > 0$, the $\mathscr Q$--balls $B^I_t$ are the empty balls defining the $\mathscr Q$--Delaunay triangulation of $X_t$.
\end{proof}

\subsubsection{The geometry of a degenerate quadratic form.}\label{sec:degen-geom}

Let $\mathscr Q$ be a degenerate quadratic form as in the previous subsection. Parallel translation of the inner product $\langle \cdot, \cdot \rangle_{\mathscr Q}$ determines a flat semi-Riemannian metric on $\RR^d$, which in this case is degenerate. The isometries of such a metric form an infinite-dimensional group, indicating that this degenerate geometry does not have very much structure. We may impose more structure by thinking of the degenerate directions as having infinitesimal length rather than zero length. Such structure is best described by a group $G$ acting on $\RR^d$ (the ``symmetries of the structure") rather than a metric. We define $G$ as a limit of the isometry groups of non-degenerate quadratic forms $Q_t$ which are degenerating to $ \mathscr Q$. 
Specifically, for each $t > 0$, define $Q_t$ by $Q_t(y,z) = Q(y,0) + tQ(0,z)$,
and let $G_t = \Isom(Q_t)$. Define $G$ to be the limit of $G_t$ as $t \to 0$ in the Chabauty topology on the space of subgroups of the Lie group $\aff(\RR^d)$ of affine transformations of $\RR^d$: an element $g \in \aff(\RR^d)$ is in $G$ if and only if there exists $g_t \in G_t$ for each $t > 0$ such that $g_t \to g$ as $t \to 0$.

\begin{Proposition} \label{pr:group}
The Lie subgroup $G \subset \aff(\RR^d)$ has the following properties.
\begin{enumerate}
\item $\dim G = \dim G_t$ for all $t > 0$. 
\item Each $g \in G$ preserves the flat degenerate metric determined by $ \mathscr Q$.
\item Each $g \in G$ preserves the fibration determined by the projection $\pi : \RR^d = \RR^m \oplus \RR^{d-m} \to \RR^m$. If $Q_V$ denotes the restriction of $Q$ to $\RR^m$, a non-degenerate form, then $g$ acts as a $Q_V$-isometry, denoted $\varpi(g)$, on $\RR^m$. The map $\varpi : G \to \Isom(Q_V)$ is a surjective homomorphism satisfying the equivariance property: $\pi(gx) = \varpi(g) \pi(x)$.\label{item:base-metric}
\item Let $Q_U$ denote the restriction of $Q$ to the $\RR^{d-m}$ factor, a non-degenerate quadratic form. Then $Q_U$ determines, by parallel translation, a flat semi-Riemannian metric on each fiber $\pi^{-1}(y) = \{(y, z): z \in \RR^{d-m}\}$. These semi-Riemannian metrics are also preserved by $G$, meaning that $g \in G$ takes the metric on $\pi^{-1}(y)$ to the metric on $g \pi^{-1}(y) = \pi^{-1}(\varpi(g)y)$.\label{item:fiber-metric}
\item $G$ is the unique subgroup of $\aff(\RR^d)$ satisfying \eqref{item:base-metric} and \eqref{item:fiber-metric}.
\end{enumerate}
\end{Proposition}

In coordinates respecting the splitting $\RR^d = \RR^m \oplus \RR^{d-m}$, we may describe the group $G$ as the collection of all affine transformations with any translational part and whose linear part has the form 
$$\begin{pmatrix} A & 0 \\ B & C \end{pmatrix}$$
where $A \in \OO(Q_V), C \in \OO(Q_U)$ and $B$ is any $(d-m) \times m$ matrix.

The geometry defined by the group $G$ is much more rigid then the geometry defined by just the degenerate flat metric associated to $\mathscr Q$. This more rigid geometry appears naturally in the study of geometric structures transitioning from flat semi-Riemannian geometry of one signature to another, see~\cite{coo_lim}. While the notion of Delaunay decomposition with respect to $\mathscr Q$ is \emph{not} preserved by $\Isom(\mathscr Q)$, it is preserved by $G$. As in the non-degenerate case, see Section~\ref{mobius} for a stronger result about the invariance of the $\mathscr Q$--Delaunay decompositions. 

\begin{Proposition}\label{invariance2}
The set of $\mathscr Q$-balls is invariant under $G$. Hence so is the notion of $\mathscr Q$-Delaunay decomposition: if $g \in G$ and $X \subset \RR^{d}$ is in spacelike and generic position with respect to $\mathscr Q$, then 
$\Del_{{\mathscr Q}}(g(X)) = g(\Del_{{\mathscr Q}}(X))$.
\end{Proposition}

\begin{proof}
We use the notations introduced in Lemma \ref{lem:converging-balls} and Proposition \ref{pr:group} above.
For any $t > 0$, a simple calculation shows that the map $L_t$ takes $Q$-balls to $Q_{t^2}$-balls (bijectively).
Hence Lemma \ref{lem:converging-balls} implies that every $\mathscr Q$-ball is a limit of $Q_{t}$-balls.  The result follows because $G_t$ preserves the set of all $Q_t$ balls so its limit group $G$ must preserve the set of limits of $Q_t$ balls which includes the $\mathscr Q$ balls.
\end{proof}

\begin{Remark}
There does not seem to be any reasonable `geometric' notion of Voronoi decomposition with respect to a degenerate quadratic form $\mathscr Q$, because the $\mathscr Q$-balls do not have a center: one should think of the center as being at infinity. However, we note that by Theorem \ref{thm:rescaled}, the combinatorics of $\Vor_{Q}(X_t)$ is constant, so it may be natural to define a combinatorial $\mathscr Q$--Voronoi decomposition to agree with the combinatorics of $\Vor_{Q}(X_t)$.
\end{Remark}

\subsection{The convex hull construction revisited}\label{sec:Hpq}

Let us now return to the convex hull construction used in the proof of Theorem~\ref{thm:Del-gen}.
First, let $Q$ be a non-degenerate quadratic form on $\RR^{d}$ of signature $(p,q)$. Consider the quadratic form $\widehat Q$ on $\RR^{d+2}$ defined by $$\widehat Q(x_1, \ldots, x_d, x_{d+1}, x_{d+2}) := Q(x_1, \ldots, x_d) -x_{d+1}x_{d+2}.$$
Then $\widehat Q$ has signature $(p+1, q+1)$, and the open subset $\mathbb X \subset \mathbb{RP}^{d+1}$ consisting of negative lines in $\RR^{d+2}$ with respect to $\widehat Q$, that is
$$\mathbb X = \{x \in \mathbb{R}^{d+2}\setminus\{0\} \mid \widehat Q(x) < 0\}/\R^*,$$ 
is a model for semi-Riemannian geometry of signature $(p+1, q)$ with constant negative curvature. The group $\PO(\widehat Q)$ of linear transformations of $\mathbb{RP}^{d+1}$ preserving $\widehat Q$ is the isometry group of a homogeneous semi-Riemannian metric of signature $(p+1,q)$ and constant negative sectional curvature. 
Now, consider the affine chart $x_{d+2} =1$. In this chart, the ideal boundary of $\mathbb X$ is described by the equation $$x_{d+1} = Q(x_1, \ldots, x_d),$$ i.e. the graph in $\RR^{d+1}$ of $Q$. Hence, via the map $x \mapsto (x, Q(x))$, $\RR^{d}$ may be regarded as an open (and dense) set on the ideal boundary $\partial \mathbb X \subset \mathbb{RP}^{d+1}$. Indeed, $\partial \mathbb X$ is the natural conformal compactification of the flat semi-Riemmanian metric on $\RR^d$ defined by $Q$, and the affine isometries $\OO(Q) \ltimes \RR^d$ of that flat metric on $\RR^d$ are naturally a subgroup of the full conformal group $\widehat G = \PO(\widehat Q)$ of $\partial \mathbb X$. Hence in the proof of Theorem~\ref{thm:Del-gen}, the convex hull $\mathscr C$ of the lifts $X'$ of $X$ to the graph of $Q$ may be regarded as a convex polyhedron in $\mathbb X$ whose vertices lie on the ideal boundary $\partial \mathbb X$. Such a polyhedron $\mathscr C$ is called an \emph{ideal polyedron}.

Next, suppose $\RR^d = \RR^m \oplus \RR^{d-m}$ is a $Q$--orthogonal decomposition, with $Q = Q_V \oplus Q_U$ and, as in Section~\ref{sec:degen}, consider the degenerate quadratic form $\mathscr{Q}$ on $\RR^d$ defined by $\mathscr{Q}(y,z) = Q_V(y)$ for all $y \in \RR^m, z \in \RR^{d-m}$. In Section~\ref{sec:degen}, we defined a subgroup $G$ of the group of affine transformations of $\RR^d$ which captures the geometry of a quadratic form thought of as having finite part the degenerate quadratic form $\mathscr{Q}$ and infinitesimal part described by $Q_U$ in the degenerate directions $\RR^{d-m}$. The group $G$ was defined to be the limit as $t \to 0$ of the isometry groups of the flat metrics defined by quadratic forms $Q_t = Q_V \oplus t Q_U$. We may similarly consider the quadratic forms $$\widehat Q_t(x_1, \ldots, x_{d+2}) = Q_V(x_1, \ldots, x_m) + t Q_U(x_{m+1}, \ldots, x_{d}) - x_{d+1}x_{d+2}.$$
We denote the limit as $t \to 0$ of the projective orthogonal groups $\PO(\widehat Q_t)$ by $\widehat G$. Then $\widehat G$ acts on the space $\mathbb X \subset \RP^{d+1}$ of lines of negative signature with respect to $\widehat{\mathscr{Q}} = \mathscr{Q} - x_{d+1}x_{d+2}$ preserving a degenerate semi-Riemannian metric defined by $\mathscr{Q}$ and also preserving a quadratic form naturally isomorphic to $Q_U$ on the degenerate subspace, a copy of $\RR^{d-m}$, of each tangent space. Note that $G$ is naturally the subgroup of $\widehat G$ that preserves the affine chart $x_{d+2} = 1$. As in the non-degenerate case above, the subset of the ideal boundary $\partial \mathbb X$ that lies in the affine chart $x_{d+2} = 1$ naturally identifies with $\RR^d$ via the map $x \mapsto (x, \mathscr{Q}(x))$, and the full ideal boundary $\partial \mathbb X$ is thought of as a conformal compactification of the degenerate geometry of $\RR^d$ described by $G$. Again, the points of $X$ in Theorem~\ref{thm:Del-gen} are thought of as points on $\partial \mathbb X$ and the convex hull $\mathscr C$ constructed in the proof of the theorem is an ideal polyhedron in $\mathbb X$. The space $\mathbb X$ and its symmetry group $\widehat G$ describe constant curvature semi-Riemannian geometry which is infinitesimal in some directions. See~\cite{coo_lim} for more on constant curvature semi-Riemannian geometries and their limits.

In dimension $d = 2$, consider the Euclidean quadratic form $Q(x_1, x_2) = x_1^2 + x_2^2$. Then $\widehat Q(x_1, x_2, x_3, x_4) = x_1^2 + x_2^2 - x_3x_4$ defines a copy $\mathbb X$ of the three-dimensional hyperbolic space $\mathbb H^3$. The Euclidean plane $\RR^2$ naturally identifies with the subset of the ideal boundary $\partial \mathbb X$ lying in the affine patch $x_{4} =1$. A set of points $X$ in spacelike and general position with respect to $Q$ determines a convex ideal polyhedron $\mathscr C$ in $\HH^3$ and conversely. 

Similarly, consider the standard Lorentzian quadratic form $Q(x_1, x_2) = x_1^2 - x_2^2$. Then $\widehat Q(x_1, x_2, x_3, x_4) = x_1^2 - x_2^2 - x_3x_4$ defines a copy $\mathbb X$ of the $2+1$ dimensional \emph{anti de Sitter space} $\AdS^3$, a model for constant negative curvature Lorentzian geometry in dimension $2+1$. The Minkowski space $\RR^{1,1}$ naturally identifies with the subset of the ideal boundary $\partial \mathbb X$ lying in the affine patch $x_{4} =1$. Again, a set of points $X$ in spacelike and general position with respect to $Q$ determines a convex ideal polyhedron $\mathscr C$ in $\AdS^3$ and conversely. 

Next, define the degenerate quadratic form $\mathscr{Q}$ from $Q$ being either the Euclidean quadratic form or the Lorentzian quadratic form above using the coordinate splitting $\RR^2 = \RR^1 \oplus \RR^1$ and let $\widehat G$ and $\mathbb X$ be defined as above. Then the degenerate plane $\RR^{1,0,1}$ naturally identifies with the subset of the ideal boundary $\partial \mathbb X$ lying in the affine chart $x_{4} =1$. Here $\mathbb X$ and its symmetry group $\widehat G$ gives what is known as three-dimensional \emph{half-pipe geometry} $\HP^3$, defined in~\cite{dan_age}. It may be regarded as the geometry of an infinitesimal neighborhood of a hyperbolic $2$-plane in either the three-dimensional hyperbolic space $\HH^3$ or the three-dimensional anti de Sitter space $\AdS^3$. Again, a set of points $X$ in spacelike and general position with respect to $\mathscr{Q}$ determines a convex ideal polyhedron $\mathscr C$ in $\HP^3$ and conversely.

In Section~\ref{sec:dihedral}, we will relate the geometry of the Delaunay triangulation of $X$ with the geometry of the corresponding ideal polyhedron $\mathscr C$ in each of the three geometric contexts above. We then characterize Delaunay triangulations in $\RR^{1,1}$ and $\RR^{1,0,1}$ by applying a recent characterization of ideal polyhedra that we obtained in~\cite{dan_pol}. The following theorem extends to $\AdS^3$ and $\HP^3$ a famous theorem of Rivin~\cite{riv_ach} about ideal polyhedra in hyperbolic three-space.

\begin{theorem}[\cite{dan_pol}, Theorem 1.3 and 1.7] \label{thm:DMS}
Let $\mathbb X$ be $\AdS^3$ or $\HP^3$. Let $\Gamma'$ be a triangulated graph on the sphere and let $w: E(\Gamma') \to \RR \setminus\{0\}$ be a weight function on the edges $E(\Gamma')$ of $\Gamma'$.
 Then there exists a convex ideal polyhedron $\mathscr C$ in $\mathbb X$ and an isomorphism of $\Gamma'$ with the one-skeleton of $\mathscr C$ which takes the weights $w$ to the dihedral angles of $\mathscr C$ if and only if the following conditions hold:
\begin{enumerate}
\item[(i)] For each vertex $v$ of $\Gamma'$, the vertex sum $\sum_{e \sim v} w(e) = 0$.
\item[(ii)] The edges $e$ for which $w(e) < 0$ form a Hamiltonian cycle in $\Gamma'$.
\item[(iii)] If $c$ is a Jordan curve on the sphere transverse to $\Gamma'$ which crosses exactly two edges of negative weight, then $$\sum_{e \in E_c} w(e)  \geq 0$$
where $E_c$ are the edges of $\Gamma'$ crossed by $c$. Further, equality occurs if and only if $c$ encircles a single vertex.
\end{enumerate}
\end{theorem}

\subsection{M\"obius invariance of $Q$--Delaunay decompositions}\label{mobius}

Propositions~\ref{invariance1} and~\ref{invariance2} state that $Q$--Delaunay decompositions are invariant under the action of the isometry group of the flat geometry associated to $Q$. In this section, we search for a larger collection of transformations leaving the $Q$--Delaunay decomposition invariant. To this end, recall that the construction of the $Q$--Delaunay decomposition for $X$ in the proof of Theorem~\ref{thm:Del-gen}, interpreted according to the observations of Section~\ref{sec:Hpq}, involves forming the convex ideal polyhedron $\mathscr C$ in the negatively curved geometry $\mathbb X$ associated to $\widehat Q$ whose vertices are the points $X$, thought of as lying in the ideal boundary $\partial \mathbb X$ using the standard chart $\RR^d \hookrightarrow \partial \mathbb X$.
The convex hull is a construction invariant under the full group of $Q$-M\"obius transformations $\widehat G = \OO(\widehat Q)$. However, the second half of the construction of $\Del_Q(X)$ involves projecting the \emph{bottom faces} $\partial \mathscr C_-$ back down to $\RR^d$ and the notion of bottom and top faces of $\mathscr C$ depends on a choice of a point $\infty \in \partial \mathbb X$ to project from. So, the $Q$--Delaunay decomposition of $X$ should be invariant under any M\"obius transformation $f \in \widehat G$ for which the bottom faces of $f(\mathscr C)$ relative to $\infty$ are precisely the image of the bottom faces of $\mathscr C$ relative to $\infty$. Let us now make this more precise.

First, we note that strictly speaking, even for small $f \in \widehat G$, $f(\Del_Q(X))$ is rarely a \emph{linear} cellulation; indeed, the image of an affine plane under $f$ is typically some $Q$-sphere (of the same dimension). However, as long as the image of a cell of $\Del_Q(X)$ is contained (and compact) in $\RR^d$, we may \emph{isotop} it relative to its vertices so that it becomes a linear cell (i.e. pull it tight). Note that even if the image of each cell of $\Del_Q(X)$ lies in $\RR^d$, the straightening of $f(\Del_Q(X))$ may fail to be a cellulation of the convex hull of $f(X)$ (and even if it is a cellulation of the convex hull, it might not be $\Del_Q(f(X))$). Henceforth we will write $f(\Del_Q(X))\sim \Del_Q(f(X))$ to mean that $\Del_Q(f(X))$ is the straightening of $f(\Del_Q(X))$. 

As a warm-up to the general case, we first discuss the Euclidean case, where the precise statement of the M\"obius
invariance (below) is presumably well-known to specialists. 
Given $X$, let $\mathcal{C}(X)$ denote the finite collection of all full and empty spheres 
(see Definition \ref{def:Del-gen} and Remark \ref{full}).

\begin{Lemma}\label{Mob_Euclid}
Let $X$ be a finite set of points in $\RR^d$, let $Q$ be the standard Euclidean quadratic form
and let $f \in \widehat G = \PO(d,1)$. Then $f(\Del_Q(X))$ is isotopic to $\Del_Q(f(X))$ if and only if
$f^{-1}(\infty)$ lies outside each sphere of the collection $\mathcal C(X)$ of empty and full spheres associated to $X$.
\end{Lemma}

\begin{Remark}
Note that the conclusion of Lemma~\ref{Mob_Euclid} may fail if the condition 
is only required to hold for the empty circles.  There are easy examples of this even with $|X| = 4$.
\end{Remark}

\begin{proof}[Outline of the proof]
Let $f \in \widehat G$ and assume that $f(X)$ lies in $\RR^d$ (no point of $X$ is sent to $\infty$). Consider a sphere $S$ in $\RR^d$. 
The condition that $S$ be empty of points of $X$ means that in the conformal compactification $\partial \mathbb X= \RR^d \cup \{\infty\}$, $\infty$ lies on one side of $S$ while $X$ lies on the other (allowing for some points of $X$ to be on $S$). Hence if $S$ is an empty sphere, then $f(S)$ is empty of points of $f(X)$ if and only if $f^{-1} (\infty)$ lies in the component of the complement of $S$ containing $\infty$, i.e. $f^{-1} (\infty)$ is outside of $S$. Hence, if $f^{-1} (\infty)$ lies outside all of the empty spheres associated to $X$, then the cells of $f(\Del(X))$ straighten to cells of $\Del(f(X))$. However, $\Del(f(X))$ could have additional top-dimensional cells, if there are additional empty spheres. Since the top-dimensional cells of $\Del(f(X))$ are projections of faces of the convex hull $f(\mathscr C)$ of the points $f(X) \subset \partial \mathbb X$ (using here that the convex hull construction in $\mathbb X$ is invariant under $\widehat G$), we need only check now that the images of the full spheres associated to $X$ are not empty for $f(X)$. This is precisely the condition that $f^{-1} (\infty)$ lies outside of every full sphere associated to $X$.
\end{proof}

The condition in Lemma \ref{Mob_Euclid} does not generalize well to the case of
non-positive definite quadratic forms, since in general a $Q$--sphere in $\partial\mathbb{X}$ does not separate 
$\partial\mathbb{X}$ making it non-sensical to ask for a point of $\partial \mathbb X$ to lie inside or outside of a $Q$-sphere. For the sake of generalizing, let us reformulate Lemma \ref{Mob_Euclid} as follows.

\begin{Lemma}\label{Mob_Euclid2}
Let $X$ be a finite set of points in $\RR^d$, let $Q$ be the standard Euclidean quadratic form, 
and let $f \in \widehat G = \PO(d,1)$. Then $f(\Del_Q(X))$ is isotopic to $\Del_Q(f(X))$ if and only if
$f^{-1}(\infty)$ is in the same connected component as $\infty$ in the complement of the union of the spheres in $\mathcal C(X)$ in the conformal compactification $\partial \mathbb X$ of $(\RR^d,Q)$.
\end{Lemma}

In the general setting, the condition above needs to be slightly expanded as follows. 
Recall from Section~\ref{sec:Hpq} that, for a general quadratic form $Q$, the conformal compactification $\partial \mathbb X$ contains many points at infinity, indeed $\partial \mathbb X = \RR^d \cup \mathscr L(\infty)$ is the disjoint union of $\RR^d$ with the light cone of a point $\infty$ at infinity. Here the light cone  $\mathscr L(x)$ of a point $x \in \partial \mathbb X$ is the union of all $y \in \partial \mathbb X$ such that $\langle x, y\rangle_{\hat Q} = 0$. If $x \in \RR^d$, then $\mathscr L(x) \cap \RR^d$ is precisely the points of $\RR^d$ of null displacement from $x$ in the $Q$ norm (this only looks like a standard cone if $Q$ is Lorentzian). The stereographic 
projection used in the proof of Theorem~\ref{thm:Del-gen} then conformally identifies the complement of $\mathscr L(\infty)$ in $\partial \mathbb X$ with $(\RR^d,Q)$.

\begin{theorem}\label{Mob}
Let $Q$ be any quadratic form on $\RR^d$ and let $X$ be a finite set of points in $\RR^d$ 
in space-like position with respect to $Q$. Let $f \in \widehat G$. Then $f(\Del_Q(X))$
is isotopic to $\Del_Q(f(X))$ if 
$f^{-1}(\infty)$ lies in the same connected component as $\infty$ in the complement in 
the conformal compactification $\partial \mathbb X$ of $(\RR^d, Q)$ of the union of 
\begin{enumerate}
\item the collection $\cC$ of all the empty and full $Q$-spheres,
\item the light cones $\mathscr L(x)$ of the points of $X$.
\end{enumerate}
\end{theorem}

The proof of Theorem \ref{Mob} uses a simple statement in projective geometry. 
We use the notation of Section \ref{sec:Hpq}.

\begin{lemma} \label{lem:Mob}
Let $S$ be a $Q$-sphere in $\partial \mathbb X$ realized as the (transverse) intersection of a hyperplane $P\subset \mathbb{RP}^{d+1}$ with $\partial \X$. Let $P_0=P\cap \X$. Then a point $x \in \RR^d \setminus S$ lies inside the $Q$-sphere $S$ if and only if the line in $\mathbb{RP}^{d+1}$ passing through $x$ and $\infty$ crosses the hyperplane $P$ inside of $P_0$.
\end{lemma}

\begin{proof}
Let $S$ be defined by the equation:
$$Q(x) = \varphi(x) + D$$
for $x \in \RR^d$, where $\varphi: \RR^d \to \RR$ is a linear functional and $D \in \RR$ is a constant. Recall the explicit embedding $x \mapsto [x: Q(x): 1] \in \RP^{d+1}$. 
Then $P$ is the set of $[y: a: b] \in \RP^{d+1}$ such that $\varphi(y) = a - Db,$ where $y \in \RR^d$ and $a,b \in \RR$.
The point $\infty = [0_d:1:0]$ (where $0_d$ is the zero vector in $\RR^d$) and let $x \in \RR^d \setminus S$. Then the line $\ell$ in $\RP^{d+1}$ passing through $\infty$ and $x$is given by $[tx: tQ(x) + s: t]$. Then $\ell$ passes through $P$ at the point $p$ where $s = t(\varphi(x) - Q(x) + D)$. Let us evaluate (the sign of) $\widehat Q$ at this point:
\begin{align*}
\widehat Q(tx, tQ(x) +s, t) &= \widehat Q(tx, tQ(x) + t(\varphi(x) - Q(x) + D), t)\\ &= \widehat Q(tx, t\varphi(x) + tD, t) \\ &= Q(tx) - (t\varphi(x) + tD)t  = t^2(Q(x) - (\varphi(x) + D)). 
\end{align*}
Hence $p \in \mathbb X$ if and only if $Q(x) < \varphi(x) + D$ if and only if $x$ is inside the $Q$-sphere~$S$.
\end{proof}

\begin{proof}[Proof of Theorem~\ref{Mob}]
Let $f_t \in \widehat G$ be a family of M\"obius transformations with $f_0 = I$ such that for all $t$, $f_t^{-1} (\infty)$ does not lie on the light cone of any point of $X$ not on any empty or full $Q$-sphere in $\mathcal C$. We must first show that $f(X)$ is space-like position, so that $\Del_Q(f(X))$ is defined. Consider the spacelike segment $\sigma_{ij}$ connecting $x_i$ to $x_j$ in $\RR^d$. Then, since $f_t$ is a conformal deformation, $f_t(\sigma_{ij})$ is a spacelike segment in the conformal metric on $\partial \mathbb X$ for all $t$. The only way $f_t(x_i)$ and $f_t(x_j)$ could fail to have spacelike relative position is if $f_t(\sigma_{ij})$ is not contained in $\RR^d$, equivalently if $f_t(\sigma_{ij})$ crosses $\mathscr L(\infty)$, equivalently if $f_t^{-1}\mathscr L(\infty) = \mathscr L(f_t^{-1} (\infty))$ crosses $\sigma_{ij}$. If this happens, then for some (possibly smaller) value of $t$, $\mathscr L(f_t^{-1} (\infty))$ contains the point $x_i$ or $x_j$, equivalently $f_t^{-1}\infty$ lies in $\mathscr L(x_i)$ or $\mathscr L(x_j)$. This does not happen by assumption.

Now, to show that $f_t(\Del_Q(X)) \sim \Del_Q(f_t(X))$, we argue as in the Euclidean case. Let $\Delta$ be a cell of $\Del_Q(f_t(X))$. Then $\Delta$ is the projection to $\RR^d$ of a face of the convex hull in $\mathbb X$ of the points $f_t(X) \subset \partial \mathbb X$. This convex hull is precisely $f_t(\mathscr C)$, hence the vertices of $\Delta$ are the images under $f_t$ of a subset of $X$ which span a top or bottom face of $\mathscr C$, and therefore define one of the empty or full spheres with respect to $X$. So it remains to show that for each empty (respectively full) sphere $S$ in $\mathcal C$, $f_t(S)$ remains empty (respectively full) with respect to $f_t(X)$. 

Consider an empty or full $Q$-sphere $S$ in $\mathcal C(X)$. We now show that $f_t(S)$ is an empty (resp. full) sphere for $f_t(X)$
for all $t$. Let $P\subset \mathbb{RP}^{d+1}$ be the hyperplane as in Lemma~\ref{lem:Mob} such that $P \cap \partial \mathbb X = S$. Let $x_i \in X$ be a point not on $S$. By the Lemma, the point $f_t(x_i)$ lies strictly inside the $Q$-sphere $f_t(S)$ if and only if the line $\ell_t$ passing through $f_t(x_i)$ and $\infty$ intersects $f_t(P)$ in the interior of $\mathbb X$. If $f_{t'}(x_i)$ is at some time $t'$ on the wrong side of $f_{t'}(S)$, then there is a time $t$ for which $\ell_t \cap f_t(P) \in \partial \mathbb X \cap f_t(P) = f_t(S)$, or equivalently the line $f_t^{-1} \ell_t$ in $\RP^{d+1}$ passing through $x_i$ and $f_t^{-1} (\infty)$ intersects $S$. Hence, since $x_i \notin S$, then either $f_t^{-1} (\infty) \in S$, which we assume does not happen, or the entire line $f_t^{-1} \ell_t$ lies in $\partial \mathbb X$ which implies that $f_t^{-1} (\infty) \in \mathscr L(x_i)$, which we also assume does not happen. Thus for any $t$, $f_t(x_i)$ lies strictly inside $f_t(S)$ if and only if $x_i$ lies strictly inside $S$. 
\end{proof}

\section{Higher signature Delaunay decompositions in dimension $d = 2$}\label{d2}

In dimension $d = 2$, there are two interesting geometric settings to consider beyond the classical case of the Euclidean plane (i.e. $Q = Q_{2,0}$ a positive definite form).
The first is the \emph{Minkowski plane}, $\RR^{1,1}$, equipped with quadratic form $Q_{1,1}(x_1, x_2) = x_1^2 - x_2^2$. See Section~\ref{sec:non-degen}.
The second is the degenerate geometry $\RR^{1,0,1}$ defined by the degenerate quadratic form $Q_{1,0,1}(x_1, x_2) = x_1^2$ as in Sections~\ref{sec:degen} and Section~\ref{sec:Hpq}.

There are two collections of angles associated to a $Q$--Delaunay triangulation $\Del_Q(X)$, the collection of \emph{interior angles} (Definition~\ref{def:interior-angles}) of the triangles and the collection of \emph{edge angles}, which are angles formed by the $Q$-circles circumscribing adjacent triangles (Definition~\ref{def:edge-angles}), assigned to the edges of $\Del_Q(X)$. We investigate the behavior of the interior angles and edge angles generalizing several important results from the Euclidean setting.

\subsection{Interior angles and edge angles of higher signature Delaunay triangulations} First, we define a notion of angle in each of the two geometries of interest.

\subsubsection{Angles in the Minkowski plane $\RR^{1,1}$}
\label{ssc:angles11}

In the Minkowski plane $\RR^{1,1}$, there are two connected components of spacelike directions. The angle formed by two spacelike unit vectors $v, w$ in the same component is defined to be the positive number $\varphi > 0$ satisfying the equation $$\cosh \varphi = \langle v, w \rangle_{Q_{1,1}}.$$
The angle formed by two spacelike unit vectors $v, w$ in opposite components is defined to be the negative number $\varphi < 0$ satisfying $$\cosh \varphi = - \langle v, w \rangle_{Q_{1,1}}.$$

\subsubsection{Angles in the degenerate plane $\RR^{1,0,1}$}

Following the notation and ideas of Section~\ref{sec:degen},
let $Q_{1,0,1} = Q_V$ and $Q_{2,0} = Q_V + Q_U$ and $Q_{1,1} = Q_V - Q_U$, 
where $Q_V(x_1, x_2) = x_1^2$ and $Q_U(x_1, x_2) =  x_2^2$.
Here we focus on the case of $Q_{1,0,1}$. 
We say that a non-zero vector $v$ is {\em non-degenerate}
if $Q_{1,0,1}(v,v)>0$.

Let $v = (v_1, v_2), w = (w_1, w_2)$ be two unit vectors with respect to $Q_{1,0,1}$, i.e. $v_1^2 = w_1^2 = 1$. Recall the map $L_t: \mathbb R^2 \to \mathbb R^2$ defined by $L_t(x_1, x_2) = (x_1, \frac{x_2}{t})$ and consider any two paths $v_t = (a_t, b_t), w_t = (c_t, d_t)$ so that $L_tv_t \to v$ and $L_t w_t \to w$ as $t \to 0$, i.e. $a_t \to v_1, b_t \to 0$, $c_t \to w_1, d_t \to 0$ and $\dot b := \frac{\D}{\D t}\big|_{t=0} b_t= v_2$ and $\dot d := \frac{\D}{\D t}\big|_{t=0} d_t= w_2$.
Then, we define the angle $\dot \theta_{vw}$ between $v$ and $w$ in $\RR^{1,0,1}$ to be the derivative
$$\dot \theta_{vw} := \frac{\D}{\D t}\big|_{t= 0} \theta_{v_t w_t},$$
where $\theta_{v_t w_t}$ denotes the $\RR^{2,0}$ or $\RR^{1,1}$ angle between $v_t$ and $w_t$.
It is easy to check that this definition does not depend on the paths of vectors $v_t, w_t$ chosen, nor does it depend on whether $\theta_{v_t w_t}$ is measured in $\RR^{2,0}$ or $\RR^{1,1}$ because in all cases
\begin{equation}\label{eqn:dot-angle}
\dot \theta_{vw} = \left\{ \begin{matrix} |v_2 - w_2| & \text{ if } v_1w_1 = 1 \\ -|v_2 - w_2| & \text{ if } v_1w_1 = -1 \end{matrix}. \right. 
\end{equation}

\begin{Proposition}
The angle $\dot \theta_{vw}$ between two directions $v,w$ is invariant under the symmetry group $G$ of $\RR^{1,0,1}$ defined in Section~\ref{sec:degen-geom}.
\end{Proposition}

\begin{proof}
Let $v,w\in \RR^{1,0,1}$, let $g\in G$, and let $v'=gv, w'=gw'$. Let $(v_t)_{t\in (0,1)}$ and 
$(w_t)_{t\in (0,1)}$ be chosen such that $L_tv_t\to v$ and $L_t w_t\to w$ as $t\to 0$, as above. 
Let $(g_t)_{t\in (0,1)}$, $g_t\in G_{t^2}$, be such that $g_t\to g$, where $G_{t^2}$ was defined in Section \ref{sec:degen-geom}. Then $g_tL_tv_t\to v', g_tL_tw_t\to w'$.
Hence since $\mathcal{L} = L_{1/t}g_tL_t$ preserves $Q$, we have:
$$ \theta_{v'w'} = \frac d{dt}_{|t=0} \theta_{\mathcal{L}v_t,\mathcal{L}w_t} = \frac d{dt}_{|t=0} \theta_{v_t,w_t}
=\theta_{vw}~. $$
\end{proof}

\begin{Proposition}
Consider a non-degenerate triangle $Q_{1,0,1}$. Then the sum of the $Q_{1,0,1}$ interior angles of the triangle is zero.
\end{Proposition}

\begin{proof}
Take the derivative of the analogous formula in Euclidean or Minkowski geometry.
\end{proof}

\begin{Proposition}\label{prop:strict-angles}
Let $p_1, p_2$ be two points with non-degenerate displacement and let $p_3, p_4$ be points in $\RR^{1,0,1}$ lying in the same connected component of the set of points which are non-degenerate with respect to both $p_1$ and $p_2$. Suppose that $p_4$ is in the interior of the triangle $\Delta p_1 p_2 p_3$.
Then the $\RR^{1,0,1}$ angles $\measuredangle p_1 p_3 p_2, \measuredangle p_1 p_4 p_2$ satisfy that 
$$\measuredangle p_1 p_3 p_2 < \measuredangle p_1 p_4 p_2.$$ 
\end{Proposition}

\begin{proof}
This follows from the definition of the $\RR^{1,0,1}$ angles in terms of limits of angles
between triples of points in $\RR^2$ (or $\RR^{1,1}$), and from the corresponding inequality in 
$\RR^2$.
\end{proof}

\subsection{Interior angles and edge angles}

Let $Q = Q_{1,1}$ or $Q_{1,0,1}$.
There are two collections of angles naturally assigned to a $Q$-Delaunay triangulation $\Del_{Q_{1,1}}(X)$. Each triangle has three \emph{interior angles}, two positive and one negative, which sum to zero.
\begin{Definition}\label{def:interior-angles}
Let $Q = Q_{1,1}$ or $Q = Q_{1,0,1}$ and let $X$ be a finite set in $\RR^2$ in spacelike and generic position with respect to $Q$. The \emph{interior angles} of $\Del_Q(X)$ is the collection of interior angles of all triangles of $\Del_Q(X)$, listed in increasing order.
\end{Definition}

Each edge $e$ of $\Del_{Q_{1,1}}(X)$ is assigned an \emph{edge angle} defined in terms of the intersection between the empty $Q$-circles circumscribing the triangles adjacent to that edge. To make a precise definition, consider two
regions $R_1, R_2$ in the plane, each with smooth boundary $S_1 = \partial R_1$, $S_2 = \partial R_2$. 
The angle formed by $R_1$ and $R_2$ at a point $p$ in the intersection $S_1 \cap S_2$ of their boundaries is defined to be the angle $\varphi$ formed by the two unit vectors $v_1, v_2$ tangent to $S_1$ and $S_2$ at $p$ and in the direction of traversal placing $R_1$ and $R_2$ on the left with respect to some (any) fixed ambient orientation. In the case that $R_1 = B_1$, $R_2 = B_2$ are $Q$--balls, the corresponding $Q$--circles $S_1$ and $S_2$ will typically intersect at two points. It is easy to check that the angle formed by $B_1$ and $B_2$ at both points is the same, making the notion of intersection angle between $Q$--circles well defined.

\begin{figure}[h]
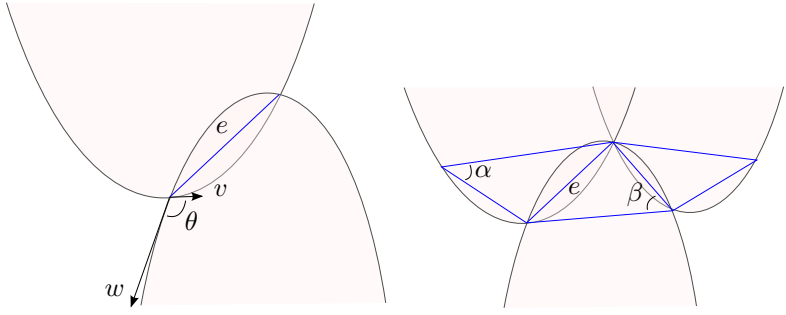

{
\centering

\def\svgwidth{5.0cm}
\input{Delaunay-edge-angles.pdf_tex}
\def\svgwidth{5.0cm}
\input{Delaunay-edge-angles-2.pdf_tex}

}
\caption{The edge angle $\theta$ associated to $e$ is the angle between the unit vectors $v$ and $w$ tangent to the two empty $Q$-circles circumscribing the triangles adjacent to $e$. Alternatively, $\theta$ is the sum of the interior angles $\alpha, \beta$ opposite to $e$.}
\end{figure}

\begin{Definition}\label{def:dihedral-angles-triangulation}\label{def:edge-angles}
Let $X$ be a finite set of points in $\RR^2$ in spacelike and general position with respect to a quadratic form $Q$ and let $\mathcal T$ be any triangulation of the convex hull of $X$. Consider an edge $e$ of $\mathcal T$. If $e$ is an interior edge, then it is adjacent to two triangles $t_1, t_2$. In this case, the \emph{edge angle} at $e$ of $\mathcal T$ is the angle $\varphi$ formed by the $Q$ balls $B_1$ and $B_2$ circumscribing $t_1$ and $t_2$ respectively. 
If $e$ is an exterior edge, then $e$ is adjacent to one triangle $t_1$. In this case, the \emph{edge angle} at $e$ is the angle between the ball $B_1$ circumscribing $t_1$ and the half-plane whose intersection with the convex hull of $X$ is precisely $e$. Typically we will denote by $\Phi: E(\mathcal T) \to \RR$ the function which assigns to each edge $e$ in the set of edges $E(\mathcal T)$ the dihedral angle $\Phi(e)$ of $\mathcal T$ at $e$.
\end{Definition}

As for Euclidean Delaunay decompositions, there is a simple relation between the interior angles
and the edge angles:

\begin{Proposition}
Let $e$ be an edge of $\Del_Q(X)$. Then the edge angle at $e$ is the sum of the interior angles opposite to $e$ in the triangles adjacent to $e$. If $e$ is an interior edge, then there are two such triangles, whereas if $e$ is an exterior edge, than there is one such triangle.
\end{Proposition}

In the case $Q = Q_{1,1}$, the proof of this property follows from an elementary geometric remark: given three points 
$a,b,c$ at time-distance $r>0$ from $0$ in $\RR^{1,1}$, then the Minkowski angles satisfy 
$\measuredangle(bac)=\measuredangle(b0c)/2$. In the case $Q = Q_{1,0,1}$, the proof proceeds by taking limits of the Minkowski or Euclidean case as usual.

\subsection{Prescribing the edge angles}\label{sec:dihedral}

The edge angles of any triangulation $\mathcal T$ as in Definition~\ref{def:edge-angles} determines a weighted planar triangulated graph. In the case that the quadratic form $Q$ is not positive definite, the following theorem characterizes exactly which weighted graphs occur as the edge angles of a $Q$--Delaunay triangulation. 

\begin{theorem}\label{thm:prescribe-angles}
Let $\Gamma$ be a triangulated graph in the plane and let $w: E(\Gamma) \to \RR \setminus\{0\}$ be a weight function on the set $E(\Gamma)$ of edges of $\Gamma$. Let $Q$ be either $Q_{1,1}$ or $Q_{1,0,1}$. Then there exists a finite set $X$ in $\RR^2$ in spacelike and general position with respect to $Q$ and an isomorphism of $\Gamma$ with the one-skeleton of the Delaunay triangulation $\Del_{Q}(X)$ which takes the weights $w$ to the edge angles of $\Del_Q(X)$ if and only if the following conditions hold:
\begin{enumerate}
\item For each interior vertex $v$ of $\Gamma$, the vertex sum $\sum_{e \sim v} w(e) = 0$.
\item The edges $e$ for which $w(e) < 0$ form a Hamiltonian path whose endpoints $v_1, v_2$ are exterior vertices.
\item The vertex sum $\sum_{e \sim v} w(e)$ associated to an exterior vertex $v$ is positive if $v = v_1$ or $v= v_2$ and is negative otherwise.
\item If $c$ is a Jordan curve in $\RR^2$ transverse to $\Gamma$ which either crosses two edges of negative weight but does not enclose $v_1$ or $v_2$, encloses $v_1$ and $v_2$ but does not cross any edge of negative weight, or crosses exactly one edge of negative weight and encloses exactly one of $v_1 $ and $v_2$, then $$\sum_{e \in E_c} w(e) - \sum_{v \in V_c } \sum_{e \sim v} w(e) \geq 0$$
where $E_c$ are the edges of $\Gamma$ crossed by $c$ and $V_c$ are the exterior vertices of $\Gamma$ enclosed inside of $c$. Equality occurs if and only if $c$ encircles a single vertex or $c$ encircles all vertices.
\end{enumerate}
\end{theorem}

\begin{proof}
First, consider a set $X$ of finitely many points in spacelike and generic position with respect to $Q$. As in the proof of Theorem~\ref{thm:Del-gen}, we lift $X$ to a point set $X'$ lying on the graph of $Q$ in $\RR^{3}$. The Delaunay triangulation $\Del_Q(X)$ is seen on the bottom of the convex hull $\mathscr C$ of $X'$. As described in Section~\ref{sec:Hpq}, the graph of $Q$ in $\RR^{3}$ is naturally an open dense subset of the ideal boundary $\partial \mathbb X$ of either $\mathbb X = \AdS^3 \subset \RP^3$ if $Q = Q_{1,1}$ or $\mathbb X = \HP^3 \subset \RP^3$ if $Q = Q_{1,0,1}$. 

Let us now add an extra vertex $v_\infty = [0:0:1:0]$ and denote the convex hull of $X'$ and $v_\infty$ by $\mathscr C_2$. Then the faces of $\mathscr C_2$ consist of the bottom faces of $\mathscr C$ (corresponding to the triangles of $\Del_Q(X)$) as well as infinite vertical walls, as viewed in $\RR^3 \subset \RP^3$ along the edges corresponding to exterior edges of $\Del_Q(X)$. It is straightforward to show that the edge angle at an edge $e$ of $\Del_Q(X)$ (Definition~\ref{def:dihedral-angles-triangulation}) is exactly the same as the dihedral angle of the corresponding edge of $\mathscr C_2$ as measured in $\mathbb X$. Indeed, the conformal semi-Riemannian structure on $\partial \mathbb X$ is compatible with the semi-Riemannian structure on $\mathbb X$. Further, since the sum of the dihedral angles around an ideal vertex of $\mathscr C_2$ are equal to zero (see \cite{dan_pol}), the dihedral angle along one of the new vertical edges $e$ of $\mathscr C_2$ is exactly minus the sum of the edge angles of the edges of $\Del_Q(X)$ incident to the vertex $v$ to which $e$ projects.

\begin{figure}[h]
{
\centering
\def\svgwidth{7.0cm}
\input{convex-hull-with-point-at-infinity.pdf_tex}
}
\caption{I.}
\end{figure}

Hence, Theorem~\ref{thm:prescribe-angles} follows directly from Theorem~\ref{thm:DMS}: Let $\Gamma'$ denote the triangulated graph on the sphere obtained from $\Gamma$ by first adding a single vertex $v_\infty$ at infinity and then connecting all exterior vertices of $\Gamma$ to $v_\infty$ with an edge. There is a one-to-one correspondence between weight functions on $\Gamma$ satisfying the conditions of Theorem~\ref{thm:prescribe-angles} and weight functions on $\Gamma'$ satisfying the conditions of Theorem~\ref{thm:DMS}. The weights on the new edges of $\Gamma'$ are determined by Condition~(i). Condition (i) for the new vertex $v_\infty$ corresponds to the case that $c$ encircles the entire polygon in Condition (4). Condition (ii) corresponds to Conditions (2) and (3). Finally, in Condition (iii), that $c$ crosses a new edge of $\Gamma'$ is equivalent to $c$ encircling the corresponding exterior vertex of $\Gamma$.
\end{proof}

\subsection{Interior angle optimization}

Consider a triangulation $\mathcal T$ of a finite point set $X$ in spacelike and general position with respect to $Q$, where $Q = Q_{1,1}, Q_{1,0,1},$ or $Q = Q_{2,0}$. The \emph{angle sequence} of $\mathcal T$ is the ordered list, sorted from smallest to largest with repetition, of the interior angles of all the triangles in $\mathcal T$. Given two triangulations $\mathcal T_1, \mathcal T_2$, we say that $\mathcal T_1$ is \emph{fatter} than $\mathcal T_2$ if the angle sequence of $\mathcal T_1$ is greater than the angle sequence of $\mathcal T_2$ with respect to the lexicographic ordering, i.e. if the first angle for which the angles sequences of $\mathcal T_1$ and $\mathcal T_2$ disagree is larger for $\mathcal T_1$.

\begin{theorem}\label{thm:angle-optimization}
Let $X$ be a finite set of points in $\RR^2$ in spacelike and general position with respect to $Q$,
where $Q = Q_{1,1}, Q_{1,0,1},$ or $Q_{2,0}$.
Then the $Q$--Delaunay triangulation $\Del_Q(X)$ is the unique fattest triangulation of $X$: the angle sequence of $\Del_Q(X)$ is larger, with respect to the lexicographic ordering, than the angle sequence of any other triangulation of $X$.
\end{theorem}

Of course, this theorem is well-known in the case $Q = Q_{2,0}$ of Euclidean geometry. As in that setting, to prove the theorem, we must first consider the simple case of a quadrilateral. To analyze that case, we need a generalization of the classical Thales' Theorem. The proof in this general setting is essentially the same.

\begin{Proposition}[Thales' Theorem]\label{prop:Thales}
Let $p_1, p_2, p_3$ be three distinct points in spacelike position lying on a $Q$--circle $S$ and let $p_4$ and $p_5$ be points on the same side of the line $\overline{p_1 p_2}$ as $p_3$ and which lie inside and outside of $S$ respectively. Further assume that $p_3, p_4, p_5$ all lie in the same connected component of the set of points that are in spacelike position with respect to $p_1$ and $p_2$. Then $\measuredangle p_1 p_4 p_2 > \measuredangle p_1 p_3 p_2 > \measuredangle p_1 p_5 p_2$.
\end{Proposition}

\begin{figure}[h]
{
\centering

\def\svgwidth{6.0cm}
\input{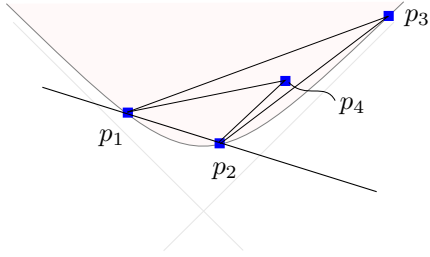}

}
\caption{In the proof of Proposition~\ref{prop:Thales}, the angle $\measuredangle p_1 p_3 p_2$ is locally constant as $p_3$ is varied, so WLOG we may assume $\Delta p_1 p_3 p_2$ contains $\Delta p_1 p_4 p_2$.}
\end{figure}

\begin{proof}
In the case that $Q$ is non-degenerate (so $Q = Q_{2,0}$ or $Q = Q_{1,1}$), the $Q$--circle has a center $O$. Observe that the triangles $\Delta O p_1 p_3$ and $\Delta O p_2 p_3$ are isosceles, from which it easily follows that the angle $\measuredangle p_1 p_3 p_2$ is locally constant as $p_3$ is varied along $S$. So without loss in generality we may assume that $\Delta p_1 p_2 p_4 \subset \Delta p_1 p_2 p_3 \subset \Delta p_1 p_2 p_5$ and the result follows.

In the case $Q = Q_{1,0,1}$, the argument is almost the same. However, in this case a $Q$--circle $S$ does not have a center (the center is at $\infty$). 
Nonetheless, the claim that the angle $\measuredangle p_1 p_3 p_2$ is locally constant as $p_3$ is varied along $S$ still holds by taking the derivatives of the same fact for $Q = Q_{2,0}$ as rescaled $Q_{2,0}$-circles converge to $S$ as in Lemma~\ref{lem:converging-balls}. So again without loss in generality we may assume that $\Delta p_1 p_2 p_4 \subset \Delta p_1 p_2 p_3 \subset \Delta p_1 p_2 p_5$, and the result follows using the strict inequality of Proposition~\ref{prop:strict-angles}.
\end{proof}

\begin{Lemma}[Angle optimization for quadrilaterals]\label{lem:quads}
Consider a convex quadrilateral $R$ whose vertices $X$ are in spacelike general position with respect to $Q$. Then, of the two triangulations of $X$, the $Q$--Delaunay triangulation $\Del_Q(X)$ is the fatter triangulation.
\end{Lemma}

\begin{proof}
The proof is nearly identical to the proof in the classical Euclidean case $Q = Q_{2,0}$.
Let $p_1, p_2, p_3, p_4$ denote the points of $X$, cyclically ordered around the boundary of $R$. Without loss in generality, the minimum (which a priori might not be unique) of the twelve angles formed by any three of the four points is $\measuredangle p_1 p_2 p_4$. Note that if $Q = Q_{1,1}$ or $Q = Q_{1,0,1}$, then $\measuredangle p_1 p_2 p_4$ will be negative. In that case, $p_3$ and $p_4$ must lie in the same connected component of the set of points in spacelike position with respect to $p_1$ and $p_2$. Since $p_3$ does not lie on the $Q$--circle $S$ passing through $p_1, p_2, p_4$, it lies either inside or outside $S$. By Proposition~\ref{prop:Thales}, $p_3$ must lie inside $S$. Hence, of the two triangulations of $R$, the one containing $\Delta p_1 p_2 p_4$ both fails to be the fatter of the two and fails to be the $Q$--Delaunay triangulation.
\end{proof}

\begin{proof}[Proof of Theorem~\ref{thm:angle-optimization}]
Let $\mathcal T$ be a triangulation of $X$ with optimal angle sequence. In other words $\mathcal T$ is the fattest triangulation or possibly tied for fattest.
Consider any quadrilateral $R$ in $\mathcal T$. We will show that the empty $Q$--circle condition holds for~$R$, meaning that each of the two triangles of $R$ has the property that its circumscribed $Q$--circle does not contain the fourth point of $R$. This is true automatically if $R$ is not convex. So assume $R$ is convex. Then the optimality of $\mathcal T$ implies that the triangulation of $R$ also optimizes the angle sequence (six angles) for $R$. Hence, by Lemma~\ref{lem:quads}, the empty $Q$--circle condition holds for $R$. 

Since each quadrilateral satisfies the empty circle condition, it follows that the polygonal surface obtained by lifting the triangles of $\mathcal T$ up to the graph of $Q$ in $\RR^3$ is locally convex along all edges. 
It is therefore a convex surface. It now follows that $\mathcal T$ must be the $Q$--Delaunay triangulation. 
\end{proof}

\section{Outline of applications}
\label{sc:5}

We now briefly outline some possible applications of higher signature Delaunay triangulations.

\subsection{Triangulations in Minkowski space}

The most direct application is to define triangulations with a given vertex
sets in spacelike position in Minkowski space, or in other constant curvature 
Lorentz spaces. 

Suppose for instance that $\Sigma\subset \R^{d-1,1}$ is a spacelike
hypersurface in the $d$-dimensional Minkowski space, and let $X$ be a 
finite set of points of $\Sigma$. The $Q$-Delaunay decomposition with 
vertex set $X$ --- where $Q$ denotes the Minkowski scalar product --- 
provides a triangulation of the convex hull of $X$ which we expect to
have good properties with respect to the ambient geometric structure.

\subsection{Analysis of discrete functions}

There are possible applications of the higher signature Delaunay decomposition not only
for sets of points in Minkoski space, but also in other situations where one
dimension (or more) plays a role fundamentally different from the others.

Consider for instance a finite set $X_0\in \R^{d-1}$ and a $k$-Lipschitz function 
$u:X_0\to \R$ obtained by sampling a function $v:\R^{d-1}\to \R$
at the points of $X_0$. After multiplying $u$ by a constant, we can suppose
that $u$ is $k$-Lipschitz, for some $k<1$. Let $X\subset \R^d$ be the graph of 
$u$. We can then consider the $Q$-Delaunay decomposition with vertex set $X$,
where 
$$ Q = dx_1^2 + \cdots + dx_{d-1}^2- dx_d^2~. $$

We expect that the combinatorics of $\Del_Q(X)$ can be useful in analysing
the points of $X_0$ where $u$ takes particularly interesting or relevant
values. Typically, points of $X$ which are vertices of ``large'' simplices
of $\Del_Q(X)$ could be particularly relevant.

For functions which are not {\em a priori} Lipschitz, it might be more 
relevant to consider the $Q'$-Delaunay decomposition with vertex set $X$,
where $Q'$ is the degenerate quadric
$$ Q' = dx_1^2 + \cdots + dx_{d-1}^2~. $$

\subsection{Proximity Graphs} 

Proximity graphs are undirected graphs in which the points spanning an edge are close in some defined sense. 

\subsubsection{Minimum Spanning Tree and Traveling Salesperson Cycle} 

Given a finite set $X \subset \R^d$, a {\em Euclidean Minimum Spanning Tree} $\mathrm{MST}(X)$ of $X$ is an (undirected) graph with 
minimum summed Euclidean edge lengths that connects all points in $X$ and which has only 
the points of $X$ as vertices. It is easy to see that it has no cycle 
(that is, that it is really a tree).
We generalize this notion to the setting of points $X$ in spacelike and general position with respect to a quadratic form $Q$ on $\RR^d$. We may measure distance between two points $x_1, x_2 \in X$ by the function $d_Q(x_1, x_2) = \sqrt{Q(x_1-x_2)}$. As in the Euclidean case, we define a \emph{$Q$-minimum spanning 
tree} $\mathrm{MST}_Q(X)$  of $X$ as a tree with spacelike edges of minimum total length 
having as vertices exactly the points of $X$. 
Any such tree satisfies the following remarkable property, well-known in the case $Q$ is the Euclidean norm.

\begin{Proposition}\label{prop:mst}
Let $X$ be in spacelike and general position with respect to the quadratic form $Q$. Then a $Q$-minimum spanning tree
$\mathrm{MST}_Q(X)$ is a sub-graph of the one-skeleton of the $Q$-Delaunay triangulation $\Del_Q(X)$.
\end{Proposition}

\begin{proof}
We need only to treat the case that $Q$ is non-degenerate. The case of $Q$ degenerate follows from Theorem \ref{thm:rescaled} by approximating $Q$ by a one-parameter family of non-degenerate quadratic forms.

Consider an edge $[x, y]$ of $\mathrm{MST}_Q(X)$ spanned by two vertices $x, y \in X$. Let $c \in [x, y]$ be the center point of the interval and consider the $Q$-ball $B$ centered at $c$ and of radius $d_Q(x, y)/2$. Specifically, $B$ is defined by the inequality $Q(z-c) < Q(y-c) = Q(x-c)$.
We show that all other points $z \in X$ lie outside of $B$. Suppose by contradiction that there is $z \in X$ such that $Q(z-c) < Q(y-c) = Q(x-c)$. We first show that $Q(z-x), Q(z-y) < Q(x-y)$. Indeed, if $Q$ is Euclidean, then this follows easily from the triangle inequality. However, we note that if $Q$ is not positive definite, the triangle inequality fails even for points in spacelike position.  Set $D^2 = Q(x-y)$. Then
\begin{align*}
D^2 &= \langle x-y, x-y \rangle_Q \\ &= Q(x) + Q(y) +2\langle x,y\rangle_Q~,
\end{align*}
and therefore 
$$ 2 \langle x, y \rangle_Q = - D^2 + Q(x) + Q(y)~. $$
Next, using that $Q\left(z-\frac{x}{2} - \frac{y}{2}\right) < \left(\frac{D}{2}\right)^2$ and the above expression for $\langle x, y\rangle_Q$, we have:
$$ Q(z) +\frac{1}{4}Q(x)+ \frac{1}{4}Q(y)-\langle z,x\rangle_Q -\langle z, y\rangle_Q + 
\frac{1}{2}\langle x, y\rangle_Q < \frac{D^2}{4} $$
so that
$$ Q(z) + \frac{1}{2}Q(x) + \frac{1}{2}Q(y) - \langle z,x\rangle_Q - \langle z,y\rangle_Q < 
\frac{D^2}{2} $$ 
and thus
$$ Q(z-x) + Q(z-y) < D^2~.$$

Hence, since both $Q(z-x), Q(z-y)$ are positive, we have that $Q(z-x), Q(z-y) < D^2 $ and so $d_Q(x,z), d_Q(y,z) < D= d_Q(x,y)$.
Next, there is some path in $\mathrm{MST}_Q(X)$ that connects $z$ to either $x$ or $y$ not passing through $[x,y]$. Without loss in generality, $z$ is connected to $y$ by such a path. Replacing the segment $[x,y]$ with $[z,x]$ yields a spanning tree with smaller total $Q$-length which contradicts that $\mathrm{MST}_Q(X)$ was minimal. 
Hence there are no points of $X$ inside the ball $B$. Hence if $X'$ denotes the lift of $X$ to the graph of $Q$, as in the proof of Theorem~\ref{thm:Del-gen}, then the plane whose intersection with the graph of $Q$ projects to $\partial B$ is a support plane to the bottom of the convex hull $\mathscr C$ of $X'$ which contains the lift $[x',y']$ of the edge $[x,y]$. Hence $[x',y']$ is an edge of the convex hull of $X'$ and so $[x,y]$ is an edge of $\Del_Q(X)$. 
\end{proof}

In the Euclidean plane, Proposition~\ref{prop:mst} has been used to provide algorithms 
calculating $\mathrm{MST}(X)$ in time $O(n\log n)$, where $n = |X|$, 
while algorithms not using the Delaunay decomposition of $X$ run in $O(n^2)$. In higher dimension, that is, if $d\geq 3$, finding an optimal algorithm remains an open problem.
For higher signatures, finding an optimal algorithm to construct a minimum spanning tree for a set
$X$ of points in spacelike position appears to be an open question.
We did not investigate whether
algorithms that apply in the Euclidean case extend to higher signatures. It might be relatively easy in the Minkowski plane, since the spacelike condition is then 
quite strong, but it is conceivable that Proposition~\ref{prop:mst} can be used for this question 
in the 3-dimensional Minkowski space.

\subsubsection{Relative neighborhood and Gabriel graph} 

There are also natural higher signature analogues of the relative neighborhood graph, introduced and studied by Toussaint~\cite{tou_the} in 1980 in the setting of the Euclidean plane, and the Gabriel graph, introduced by Gabriel--Sokal~\cite{gab_ane} in 1969 in the setting of Euclidean space of any dimension.
Let $X$ be a finite set in $\RR^d$ which is in spacelike position. Then
the {\em relative neighborhood graph} $\mathrm{RNG}_Q(X)$ of $X$ is the 
(undirected) graph which has an edge connecting $x$ to $y$ if $d_Q(x,y) < \max\{d_Q(z,x), d_Q(z,y)\}$ for all $z \in X \setminus\{x,y\}$.  

Next assume $Q$ is non-degenerate. The \emph{Gabriel graph} $\mathrm{GG}_Q(X)$ of $X$ has an edge $[x,y]$ for $x,y \in X$ if the open $Q$-ball whose diameter is the segment $[x,y]$ is empty of all other points of $X$. This notion does not seem to have an analogue for $Q$ degenerate. It is straightforward from the proof of Proposition~\ref{prop:mst} to show:
\begin{Proposition} In the case $Q$ non-degenerate,
$$\mathrm{MST}_Q(X) \subset  \mathrm{RNG}_Q(X) \subset \mathrm{GG}_Q(X) \subset \Del_Q(X).$$
In the case $Q$ degenerate,
$$\mathrm{MST}_Q(X) \subset  \mathrm{RNG}_Q(X) \subset \Del_Q(X).$$
\end{Proposition}

\subsection{Minimizing interpolation error} 

The Delaunay triangulation $\Del_{Q}(X)$ can also be characterized from a functional approximation point of view. For example, in the Euclidean setting, Lambert \cite{lam_the} showed that the classical Delaunay triangulation maximizes the arithmetic mean of the radius of inscribed circles of the triangles, while Rippa \cite{rip_min} proved that it minimizes the integral of the squared gradient.

Let $\Omega \subset \R^d$ be a bounded domain, let $T$ be a triangulation of $\Omega$, and let $f$ be a function defined on $\Omega$. We define 
$$\mathcal{Q}(T, f, p) = \|f-\hat{f}_{T}\|_{Q, L^p(\Omega)}~,$$ 
where $\hat{f}_{T}$ is the linear interpolation of $f$ based on the triangulation $T$ of $\Omega$, and where we use the quadratic form $Q$. 
The following result, due to Chen--Xu \cite{che_opt} in the Euclidean setting, characterizes the $Q$-Delaunay triangulation as the optimal triangulation for the piecewise linear interpolation over a given point set $X$ of the function $Q(x)$. The proof of Chen--Xu, which uses the convex hull interpretation of the Delaunay triangulation,  is easily generalized to the higher signature setting under the additional hypothesis that the points are in spacelike position.

\begin{Theorem}
Let $X\subset \R^d$ be a finite set in spacelike position with respect to $Q$, then
$$\mathcal{Q}(\Del_{Q}(X), Q(\cdot), p) = \min_{T\in \mathcal{P}(X)}\mathcal{Q}(T, Q(\cdot), p),\text{ for } 1\leq p\leq \infty~,$$ 
where $\mathcal{P}(X)$ is the set of all triangulations that have $X$ as vertices and $\Omega$ is the convex hull of $X$.
\end{Theorem}

\subsection{Other applications} 

There are many more other well-known and important applications/generalizations of the classical 
Delaunay triangulation, for example the notion of $\alpha$--shapes \cite{ede_ont} or 
$\beta$--skeletons \cite{kir_afr}, which have important applications in pattern recognition, 
digital shape sampling and processing, and structural molecular biology, among others. 
We encourage the reader to investigate generalizations of such constructions to the 
higher signature setting as needed.

\subsection{Implementation}

There are several efficient algorithms that can be used to compute the Delaunay triangulation
of a set of points in the Euclidean plane or in higher-dimensional space, for instance
an incremental flipping algorithm \cite{ede_sha} or a divide-and-conquer algorithm
\cite{cig_mon}, both 
running in $O(n\log(n))$ in the plane (see \cite{boi_dev} for efficient implementations).
We have not investigated the extent to which these algorithms generalize readily to the setting of $Q$--Delaunay triangulations, nor have we investigated whether similar efficiency can be achieved. Of course any convex hull algorithm may applied to directly calculate $Q$--Delaunay decompositions as in the proof of Theorem~\ref{thm:Del-gen}.

In an online appendix to this 
paper\footnote{\texttt{http://math.uni.lu/schlenker/programs/highersign/highersign.html}},
we provide a crude implementation of the computation of Delaunay decomposition for the 
Euclidean, Minkowski and degenerate bilinear forms in $\R^3$. This implementation is 
provided for illustration only, and is not intended to be computationally efficient.

\providecommand{\bysame}{\leavevmode\hbox to3em{\hrulefill}\thinspace}
\providecommand{\MR}{\relax\ifhmode\unskip\space\fi MR }
\providecommand{\MRhref}[2]{
  \href{http://www.ams.org/mathscinet-getitem?mr=#1}{#2}
}
\providecommand{\href}[2]{#2}


\begin{thebibliography}{10}

\bibitem{boi_dev}
Jean-Daniel Boissonnat, Olivier Devillers, Sylvain Pion, Monique Teillaud, and
  Mariette Yvinec, \emph{Triangulations in {CGAL}}, Comput. Geom. \textbf{22}
  (2002), no.~1-3, 5--19, 16th ACM Symposium on Computational Geometry (Hong
  Kong, 2000). \MR{1893651}

\bibitem{bro_vor}
Kevin~Q Brown, \emph{Voronoi diagrams from convex hulls}, Inf. Process. Lett.
  \textbf{9} (1979), no.~5, 223--228.

\bibitem{che_opt}
Long Chen and Jin-chao Xu, \emph{Optimal {D}elaunay triangulations}, J. Comput.
  Math. \textbf{22} (2004), no.~2, 299--308, Special issue dedicated to the
  70th birthday of Professor Zhong-Ci Shi. \MR{2058939 (2005c:41042)}

\bibitem{cig_mon}
Paolo Cignoni, Claudio Montani, and Roberto Scopigno, \emph{Dewall: A fast
  divide and conquer delaunay triangulation algorithm in $e^d$}, Computer-Aided
  Design \textbf{30} (1998), no.~5, 333--341.

\bibitem{coo_lim}
Daryl Cooper, Jeffrey Danciger, and Anna Wienhard, \emph{Limits of geometries},
  preprint arXiv:1408.4109 (2014).

\bibitem{dan_age}
Jeffrey Danciger, \emph{A geometric transition from hyperbolic to anti-de
  {S}itter geometry}, Geom. Topol. \textbf{17} (2013), no.~5, 3077--3134.
  \MR{3190306}

\bibitem{dan_pol}
Jeffrey Danciger, Sara Maloni, and Jean-Marc Schlenker, \emph{Polyhedra
  inscribed in a quadric}, preprint arXiv:1410.3774 (2014).

\bibitem{del_sur}
Boris Delaunay, \emph{Sur la sph\`ere vide. {A} la m\'emoire de {G}eorge
  {V}oronoi.}, Bull. Acad. Sc. USSR, Ph. Ser. \textbf{6} (1934), 793--800.

\bibitem{ede_sha}
H.~Edelsbrunner and N.~R. Shah, \emph{Incremental topological flipping works
  for regular triangulations}, Algorithmica \textbf{15} (1996), no.~3,
  223--241. \MR{1368251 (96j:65161)}

\bibitem{ede_ont}
Herbert Edelsbrunner, David~G. Kirkpatrick, and Raimund Seidel, \emph{On the
  shape of a set of points in the plane}, IEEE Trans. Inform. Theory
  \textbf{29} (1983), no.~4, 551--559. \MR{713690 (84m:52016)}

\bibitem{gab_ane}
K~Ruben Gabriel and Robert~R Sokal, \emph{A new statistical approach to
  geographic variation analysis}, Syst. Zool. \textbf{18} (1969), no.~3,
  259--278.

\bibitem{kir_afr}
David~G Kirkpatrick and John~D Radke, \emph{A framework for computational
  morphology}, Computational Geometry, Machine Intelligence and Pattern
  Recognition 2, 1985.

\bibitem{lam_the}
Timothy Lambert, \emph{The {Delaunay} triangulation maximizes the mean
  inradius}, Proc. 6th Canad. Conf. Comput. Geom., 1994, pp.~201--206.

\bibitem{rip_min}
Samuel Rippa, \emph{Minimal roughness property of the {D}elaunay
  triangulation}, Comput. Aided Geom. Design \textbf{7} (1990), no.~6,
  489--497. \MR{1079398 (91j:65018)}

\bibitem{riv_ach}
Igor Rivin, \emph{A characterization of ideal polyhedra in hyperbolic
  {$3$}-space}, Ann. of Math. (2) \textbf{143} (1996), no.~1, 51--70.
  \MR{1370757 (96i:52008)}

\bibitem{tou_the}
Godfried~T. Toussaint, \emph{The relative neighbourhood graph of a finite
  planar set}, Pattern Recognition \textbf{12} (1980), no.~4, 261--268.
  \MR{591314 (83h:52001)}

\end{thebibliography}
\end{document}